\documentclass[nonacm, acmsmall]{acmart}

\usepackage{xspace}

\usepackage{layouts}

\usepackage{wasysym}

\newcommand{\etal}{et~al.}

\DeclareFontFamily{OT1}{pzc}{}
\DeclareFontShape{OT1}{pzc}{m}{it}{<-> s * [1.10] pzcmi7t}{}
\DeclareMathAlphabet{\mathpzc}{OT1}{pzc}{m}{it}

\usepackage[
algo2e,
algosection,
ruled,vlined,
linesnumbered,
norelsize]{algorithm2e}
\DontPrintSemicolon
\SetAlCapHSkip{0pt}
\IncMargin{-\parindent}
\SetProgSty{}
\SetFuncSty{}
\SetArgSty{}

\usepackage{multirow}
\usepackage{afterpage}
\usepackage{array}
\usepackage{numprint}
\usepackage{amsmath,amsfonts}%
\usepackage{wasysym,pifont,marvosym}
\usepackage{nicefrac}

\usepackage{etoolbox,siunitx}
\robustify\bfseries

\definecolor{fuchsiapink}{rgb}{1.0, 0.47, 1.0}

\usepackage[normalem]{ulem}
\definecolor{removered}{rgb}{0.7, 0.0, 0.0}

\newcommand{\ceil}[1]{\ensuremath{\lceil #1 \rceil}}

\SetKwComment{tcp}{{\upshape{}/\!/}\,}{}%
\SetCommentSty{textit}

\SetKwProg{Function}{Function}{}{}
\SetKw{Continue}{continue}
\SetKwFunction{Activate}{Activate}
\SetKwFunction{rate}{rate}
\SetKwFunction{coarsen}{coarsen}
\SetKwFunction{DeltaGainUpdate}{DeltaGainUpdate}
\SetKwFunction{UpdateGainCache}{UpdateGainCache}

\newcommand{\Oh}[1]{\ensuremath{\mathcal{O}(#1)}\xspace}

\newcommand{\neighbors}{\ensuremath{\mathrm{\Gamma}}}%
\newcommand{\incnets}{\ensuremath{\mathrm{I}}}%
\newcommand{\pinsinpart}{\ensuremath{\mathrm{\Phi}}}
\newcommand{\adjblocks}{\ensuremath{\mathrm{B}}}
\newcommand{\con}{\ensuremath{\lambda}}
\newcommand{\conset}{\ensuremath{\Lambda}}
\newcommand{\frompart}{\ensuremath{V_{\text{\upshape from}}}}
\newcommand{\topart}{\ensuremath{V_{\text{\upshape to}}}}

\newcommand{\splitatcommas}[1]{%
  \begingroup
  \begingroup\lccode`~=`, \lowercase{\endgroup
    \edef~{\mathchar\the\mathcode`, \penalty0 \noexpand\hspace{0pt plus 1em}}%
  }\mathcode`,="8000 #1%
  \endgroup
}

\newcommand{\Partition}{\ensuremath{\mathrm{\Pi}}}%
\newcommand{\ocut}{\ensuremath{\mathfrak{f}_c(\Partition)}}%
\newcommand{\ocon}{\ensuremath{\mathfrak{f}_\lambda(\Partition)}}%

\newcommand{\maxsize}[1]{\ensuremath{\Delta_{#1}}}

\newcommand{\meddeg}{\ensuremath{\tilde{d}(v)}}
\newcommand{\medsize}{\ensuremath{|\tilde{e}|}}

\newcommand{\plusplus}{\texttt{++}}
\newcommand{\Cpp}[1]{C\plusplus#1}
\newcommand{\gpp}[1]{g\plusplus#1}

\newif\ifpdfplots
\pdfplotstrue

\usepackage{pgfplots}
\usepgfplotslibrary{external}
\AtBeginDocument{%
  \providecommand\BibTeX{{%
    \normalfont B\kern-0.5em{\scshape i\kern-0.25em b}\kern-0.8em\TeX}}}

\setcopyright{acmcopyright}
\copyrightyear{2020}
\acmYear{2020}
\acmDOI{10.1145/1122445.1122456}

\acmJournal{JEA}
\acmVolume{42}
\acmNumber{42}
\acmArticle{42}
\acmMonth{12}

\begin{document}

\title{High-Quality Hypergraph Partitioning}

\author{Sebastian Schlag}
\orcid{0000-0003-1550-882X}
\affiliation{%
  \institution{Karlsruhe Institute of Technology}
  \department{Institute of Theoretical Informatics}
  \streetaddress{Postfach 6980}
  \postcode{76128}
  \city{Karlsruhe}
  \country{Germany}
}
\email{research@sebastianschlag.de}

\author{Tobias Heuer}
\affiliation{%
  \institution{Karlsruhe Institute of Technology}
  }
\email{tobias.heuer@kit.edu}

\author{Lars Gottesbüren}
\affiliation{%
  \institution{Karlsruhe Institute of Technology}
}
\email{lars.gottesbueren@kit.edu}

\author{Yaroslav Akhremtsev}
\affiliation{%
  \institution{Karlsruhe Institute of Technology}
  }
\email{yaroslav.akhremtsev@kit.edu}

\author{Christian Schulz}
\orcid{0000-0002-2823-3506}
\affiliation{%
    \institution{ Heidelberg University}
  \department{Faculty for Mathematics and Computer Science}
  \postcode{69120}
  \city{Heidelberg}
  \country{Germany}
}
\email{christian.schulz@informatik.uni-heidelberg.de}

\author{Peter Sanders}
\affiliation{%
  \institution{Karlsruhe Institute of Technology}
}
\email{sanders@kit.edu}


\begin{abstract}
Hypergraphs are a generalization of graphs where edges (aka
\emph{nets}) are allowed to connect more than two vertices. They have
a similarly wide range of applications as graphs.  This paper
considers the fundamental and intensively studied problem of
\emph{balanced hypergraph partitioning}, which asks for partitioning
the vertices into $k$ disjoint blocks of bounded size while minimizing
an objective function over the hyperedges. Here, we consider the two
most commonly used objectives: the \emph{cut-net metric} and the
\emph{connectivity metric}.

We describe our open source hypergraph partitioner \emph{KaHyPar} which is based
on the successful multi-level approach -- driving it to the extreme of
using one level for (almost) every vertex. Using carefully designed
data structures and dynamic update techniques, this approach turns out
to have a very good time--quality tradeoff. We present two
preprocessing techniques -- \emph{pin sparsification using locality
  sensitive hashing} and \emph{community detection based on the
  Louvain algorithm}. The community structure is used to guide the \emph{coarsening
process} that incrementally contracts vertices.  \emph{Portfolio-based
  partitioning} of the contracted hypergraph then already achieves a good
initial solution.  While reversing the contraction process, a
combination of several refinement techniques achieves a good final
partitioning. In particular, we support a \emph{highly-localized local
  search} that can directly produce a $k$-way partitioning and
complement this with \emph{flow-based techniques} that take a more
global view. Optionally, a \emph{memetic algorithm} evolves a pool of
solution candidates to an overall good solution.

We evaluate KaHyPar for a large set of instances from a wide range of
application domains.  With respect to quality, KaHyPar outperforms all previously
considered systems that can handle large hypergraphs such as hMETIS,
PaToH, Mondriaan, or Zoltan. Somewhat surprisingly, to some extend, this
even extends to graph partitioners such as KaHIP when considering the
special case of graphs. KaHyPar is also faster than most of these
systems except for PaToH which represents a different speed--quality
tradeoff.
\end{abstract}


\maketitle
\small{\textbf{Supplementary Material:} The source code and experimental data are available online at \url{https://github.com/kahypar/kahypar} and \url{https://github.com/kahypar/experimental-results/tree/master/jea20}.}
\vfill\pagebreak
\section{Introduction}

\emph{Graphs} are a well-known and universal abstraction that models
objects (vertices) and their relations (edges). \emph{Hypergraphs} are
a slightly less well-known but similarly useful generalization where a
\emph{hyperedge} or \emph{net} may connect more than two vertices. A
fundamental problem in processing graphs and hypergraphs is to
partition the vertices into a specified number $k$ of disjoint
\emph{blocks} such that each block has bounded size and such that few nets connect vertices from more
than one block.  We consider the two most widely used objective
functions: The \emph{cut-net metric} counts the number of nets that
connect vertices from more than one block.  The \emph{connectivity
  metric} weights each net $e$ with a factor $\lambda(e)-1$, where
$\lambda(e)$ is the number of different blocks connected through
$e$. Moreover, we also consider generalizations with weighted vertices
and nets.

This \emph{balanced hypergraph partitioning (BHP)} problem has numerous
applications. Perhaps the
most well-known applications are in VLSI design \cite{DAlpert} (see also Figure~\ref{fig:hgp_example}) and in
minimizing communication costs when processing graphs or sparse
matrices in parallel \cite{DBLP:journals/tpds/CatalyurekA99,DBLP:journals/siamsc/UcarA04}.  Further applications include
storage sharding in distributed databases \cite{DBLP:journals/pvldb/CurinoZJM10,Kumar2014,DBLP:journals/pvldb/KabiljoKPPSAP17}, social network
analysis \cite{DBLP:journals/sigmetrics/HeintzC14,8790188,HyperXGP}, warehouse management, and route planning~\cite{delling_et_al:OASIcs:2017:7896}.

Since the BHP problem is NP-hard and hard to approximate
\cite{Lengauer:1990,DBLP:journals/ipl/BuiJ92}, solving it for large hypergraphs requires heuristics.  In
partitioners that achieve high quality, the \emph{multi-level
  approach}, which is illustrated in Figure~\ref{fig:multilevel_example}, takes a central role: The input is first
\emph{coarsened} by successively identifying sets of vertices that are
to be \emph{contracted}, i.e., these sets are replaced by a single
vertex in the new hypergraph. Nets that contain several of the
contracted vertices get correspondingly smaller, or can be removed entirely if
they shrink to size one. Parallel nets can be replaced by a single net
with correspondingly larger weight.  Note that any feasible partition
of the contracted hypergraph induces a feasible partition of the
original input with the same objective function value. Hence, if we
carefully contract only vertex sets that should be in the same block of
a good partition, we already get good solutions by partitioning the
hypergraph in any of its coarser representations.  In particular,
\emph{initial partitioning} can apply expensive algorithms if the
hypergraph has been shrunk to sufficiently small size.  Mistakes made due
to vertex sets that should not have been contracted can be straightened
out in an \emph{uncoarsening and refinement} phase where the
contractions are undone and optimization approaches like local search
are applied to further improve the solution.

\begin{figure}
  \centering
  \includegraphics[]{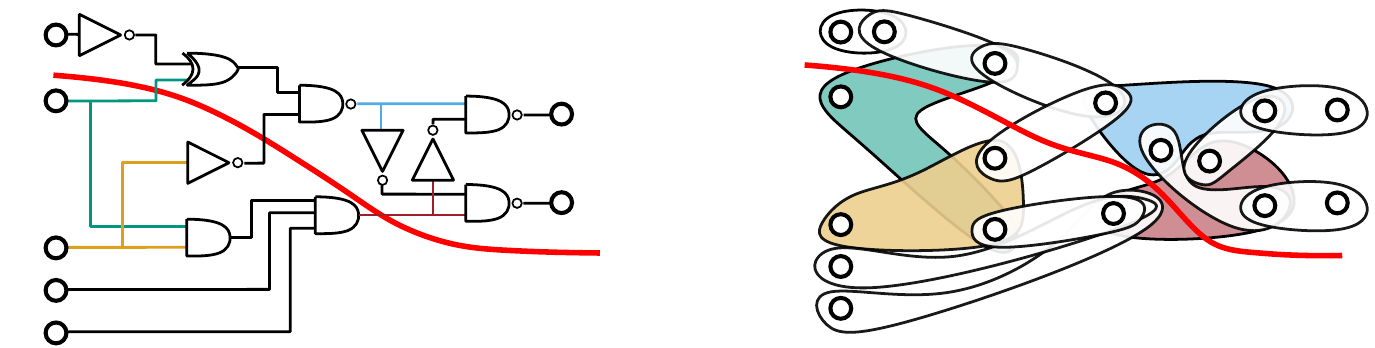}
  \caption{A circuit represented as a hypergraph together with a bipartition cutting three nets.}\label{fig:hgp_example}
\end{figure}

\begin{figure}
  \centering
  \includegraphics[width=\textwidth]{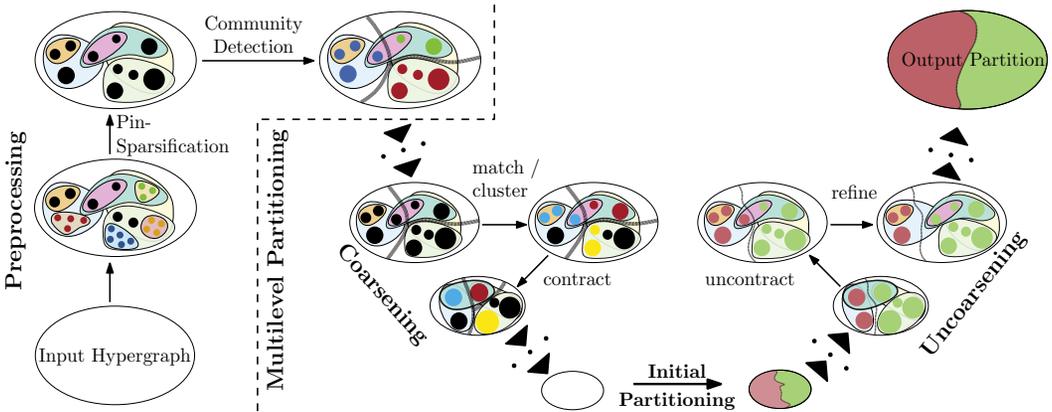}
  \caption{Illustration of several preprocessing techniques and the multi-level partitioning process.}\label{fig:multilevel_example}
\end{figure}

\subsection{Our Contributions}
The multi-level framework can be instantiated in many ways. We
contribute to all three main components (coarsening, initial
partitioning, refinement), to pre- and postprocessing, and to techniques
for experimental evaluation.  Describing all this in minute detail
requires more space than appropriate for a single journal
paper. Hence, we restrict ourselves on conveying the main ideas on a
high level and how they work together. More details can be found
primarily in the doctoral thesis of Sebastian Schlag \cite{Schlag20}%
\footnote{Text from this thesis was also used as the starting for the descriptions made here so that some short verbatim passages may stem from there.}
  and,
to a lesser extend, in the doctoral thesis of Yaroslav Akhremtsev (pin
sparsifier) \cite{ThesisYaroslav}, in the Bachelor theses of Tobias
Heuer (initial partitioning) and Robin Andre (the memetic algorithm)
\cite{AndreBA}, and in the Master thesis of Tobias Heuer (flows)
\cite{HeuerMA}.
Lars Gottesb\"uren contributed the HyperFlowCutter approach \cite{flowcutter, hfc_esa, whfc_sea}.
Preliminary descriptions of some of these results have also been
published in several conference papers
\cite{KaHyPar-R,KaHyPar-K,KaHyPar-E,KaHyPar-MF-SEA,KaHyPar-CA}.  For
one variant of the flow-based techniques, there is also a journal
paper explaining it in detail \cite{KaHyPar-MF-JEA}.  We see the added
value of our paper in giving a succinct presentation of the overall
system and how it leads to the currently best code for high-quality
hypergraph partitioning. Among others, the partitioning system described in this
paper is used for circuit partitioning~\cite{DBLP:conf/iccad/NetoATATG19}, in a graph-based algebraic modeling framework~\cite{Jalving2019}, as well as for finding high-quality contraction paths for tensor networks~\cite{gray2021hyper} for quantum circuit simulation, which is also used in a different framework~\cite{huang2020classical}.

After introducing notation and basic techniques in
Section~\ref{s:prelim}, Section~\ref{s:algorithm} explains our
algorithmic approaches.  KaHyPar represents hypergraphs as a bipartite
graph where vertices correspond to one side of the graph and nets to
the other side. The dynamic data structure described in
Section~\ref{sec:hypergraph_ds} is carefully designed to support
efficient (un)contractions.  KaHyPar supports both direct $k$-way
partitioning and recursive bipartitioning.
Section~\ref{sec:rb_vs_kway} explains this tradeoff.

Before starting the
multi-level optimization proper, KaHyPar offers two kinds of
preprocessing. A \emph{pin sparsifier} described in
Section~\ref{ssec:sparsification} contracts vertices with similar
neighborhood. This accelerates subsequent computations by reducing the
average net size. Vertices with similar neighborhood are identified
using locality-sensitive min-hashing schemes. Quality is further improved by
\emph{identifying communities} of densely connected vertices. The
approach described in Section~\ref{ssec:community_detection} applies
the Louvain algorithm for modularity clustering \cite{Louvain} on the bipartite graph
representation.  In order to derive useful communities of the hypergraph
from these clusters, appropriate weights are defined for the graph
edges.

The community structure later guides the \emph{coarsening
  phase} described in Section~\ref{sec:coarsening} by only allowing contractions inside clusters.
However, rather than directly contracting clusters, KaHyPar adapts the \emph{$n$-level
  approach} previously used by the KaSPar graph partitioner
\cite{DBLP:conf/esa/OsipovS10}. The idea is to contract only two
vertices at a time in order to obtain very fine-grained hierarchy
information that leads to improved solutions and smaller search spaces
in the refinement phase.

The initial partitioning algorithm described in
Section~\ref{sec:initial_partitioning} is based on recursive
bipartitioning and was originally developed in the Bachelor thesis of
Tobias Heuer \cite{HeuerBA}.  The algorithm recursively and repeatedly
applies a portfolio of several simple bipartitioning algorithms. These
include random initialization, BFS, local search improvement, greedy
hypergraph growing \cite{DBLP:journals/tpds/CatalyurekA99}, and an
adaption of label propagation
\cite{DBLP:journals/heuristics/MeyerhenkeS016,HenneMA} to
hypergraphs. Direct $k$-way algorithms that we also tried did not work
sufficiently well to include them into the portfolio. This is somewhat
surprising because in the refinement phase described in
Section~\ref{sec:refinement}, direct $k$-way algorithms are superior.
These local search algorithms are highly localized because they
explore the search space around a single pair of uncontracted
vertices. This makes it more likely to find nontrivial overall
improvements as already observed for graph partitioning
\cite{kaffpa,DBLP:phd/dnb/Schulz13a,DBLP:conf/esa/OsipovS10}. The
added difficulty in hypergraph partitioning is that moving a single
vertex between two blocks may on the one hand have no effect on the
objective function value but may potentially affect the desirability
of moving many other vertices.  We address these issues by developing
an \emph{improved stopping rule} for these local searches and by
reducing update overhead after vertex moves by using a \emph{gain-cache} data structure.  This
local $k$-way view is complemented by a global 2-way view using
maximum-flow computations between pairs of blocks in
Section~\ref{ssec:flows}.  This generalizes an
approach from graph partitioning \cite{kaffpa} to hypergraphs. The
new formulation even helps for the graph case since it allows more vertex moves and uses a more general tradeoff between balance and cut size~\cite{KaHyPar-MF-JEA}.

All components of KaHyPar use randomized tie breaking whenever
possible. This has the effect to increase variance of the solution
quality thus allows improved solutions by simply taking the best
from several solution attempts. Even higher quality can be achieved by
evolving a population of solutions using a memetic algorithm. A
memetic algorithm is a genetic algorithm that also uses local
search. Our variant described in
Section~\ref{subs:memetic_algo} includes recombination operators with
more than two parents and ones that guarantee that the offspring is no
worse than the parents. ``Local search'' in this context has a very
coarse-grained meaning.  It means applying
\emph{V-cycles}, i.e., applying $n$-level coarsening and
refinement. KaHyPar's V-cycles restrict contractions to blocks of the
previous solution and use either the old or a newly computed solution as the initial partition.

Section~\ref{s:experiments} then reports the main results of an
extensive experimental evaluation using close to 4000 problem instances.  Considering the algorithmic
components of KaHyPar, the pin sparsifier yields speedups on some
difficult instances while leading to a negligible deterioration on
quality.  Conversely, community-aware coarsening leads to a
significant improvement of quality for a negligible time
overhead. Flows lead to an even larger quality improvement at the cost
of increasing running time by about a factor of two.  Evolutionary
techniques yield an addional (somewhat smaller) improvement at very
large expense (but are still more efficient than brute force
repetitions).

We also compare KaHyPar with a representative set of seven
state-of-the art other hypergraph partitioners (two variants of hMetis
\cite{hMETIS-Software}, two variants of PaToH \cite{PaToH-Software},
one variant of Zoltan \cite{ZoltanAlgD-Software}, Mondriaan
\cite{Mondriaan-Software}, and HYPE~\cite{HYPE-Software}).  KaHyPar
achieves the highest quality among all these systems while being
faster than the partitioner that previously achieved the highest
quality (hMetis).  The variants of PaToH are considerably faster than
KaHyPar at the expense of worse solutions (typically around 7 \%
larger objective function value).  The only non-multi-level system,
HYPE, is sometimes yet a bit faster but at the price of an order of magnitude
larger objective function values.
We also compare KaHyPar with the state-of-the art high-quality \emph{graph} partitioner
KaFFPa. Somewhat surprisingly, KaHyPar is both faster and achieves higher quality.
We summarize our results in
Section~\ref{s:conclusions} and outline possible directions for future
work.

\subsection{Related Work}

There is extensive previous work on graph and hypergraph partitioning.
For a detailed account, we therefore refer to existing survey
articles~\cite{buluc2016recent,GPOverviewBook,DAlpert,DPapa2007,DBLP:reference/bdt/0003S19} and the thesis of
Sebastian Schlag~\cite{Schlag20} which contains a 60 page literature
survey and 45 pages of bibliography. We cite directly relevant
literature within the respective sections.
Historically, many variants of local search have been investigated
since the 1960s.  Combining local search with two-level algorithms
allowed significant quality improvements since the mid 1970s.  Since
the late 1990s, multi-level algorithms have established themselves as
the method of choice for achieving high quality.

Well-known \emph{sequential} multi-level systems with certain distinguishing characteristics include
tools originating from scientific computing like PaToH~\cite{DBLP:journals/tpds/CatalyurekA99}, originating from VLSI design like hMetis~\cite{KarypisAKS99,DBLP:journals/vlsi/KarypisK00} and MLPart~\cite{DBLP:conf/aspdac/CaldwellKM00}, and systems targeted at partitioning sparse rectangular matrices like Mondriaan~\cite{DBLP:journals/siamrev/VastenhouwB05}.
UMPa~\cite{DBLP:conf/dimacs/CatalyurekDKU12} supports directed hypergraphs and multiple objective functions.
kPaToH~\cite{DBLP:journals/jpdc/AykanatCU08} can handle multiple constraints as well as fixed vertices.

\emph{Distributed} HGP systems include Zoltan~\cite{DBLP:conf/ipps/DevineBHBC06} and Parkway~\cite{DBLP:journals/jpdc/TrifunovicK08}, and SHP~\cite{DBLP:journals/pvldb/KabiljoKPPSAP17}.
However, these currently cannot achieve the quality of sequential multi-level partitioners.
Recently, \emph{shared-memory} partitioners have made progress -- see Section~\ref{s:conclusions}.

\section{Preliminaries}\label{s:prelim}
\subsection{Notation and Definitions}
\paragraph{Hypergraphs \& Graphs}
A \textit{weighted undirected hypergraph} $H=(V,E,c,\omega)$ is defined as a set of $n$ vertices $V$ and a
set of $m$ hyperedges/nets $E$ with vertex weights $c:V \rightarrow \mathbb{R}_{>0}$ and net
weights $\omega:E \rightarrow \mathbb{R}_{>0}$, where each net $e$ is a subset of the vertex set $V$ (i.e., $e \subseteq V$).
The vertices of a net are called \emph{pins}.
We extend $c$ and $\omega$ to sets in the natural way, i.e., $c(U) :=\sum_{v\in U} c(v)$ and $\omega(F) :=\sum_{e \in F} \omega(e)$.
A vertex $v$ is \textit{incident} to a net $e$ if $v \in e$. $\mathrm{I}(v)$ denotes the set of all incident nets of $v$.
The set $\neighbors(v) := \{ u~|~\exists e \in E : \{v,u\} \subseteq e\}$ denotes the neighbors of $v$.
The \textit{degree} of a vertex $v$ is $d(v) := |\mathrm{I}(v)|$.
 We assume hyperedges to be sets rather than multisets, i.e., a vertex can only be contained in a hyperedge \emph{once}.
Nets of size one are called \emph{single-vertex} nets. We call two nets  $e_i$ and $e_j$ \emph{parallel} if $e_i = e_j$.
Given a subset $V' \subset V$, the \emph{subhypergraph} $H_{V'}$ is defined as $H_{V'}:=(V', \{e \cap V'~|~e \in E : e \cap V' \neq \emptyset \})$.

Let $G=(V,E,c,\omega)$ be a weighted (directed) graph. We use \emph{vertices}
and \emph{hyperedges/nets} when referring to hypergraphs and \emph{nodes} and \emph{edges} when referring to graphs.
However, we use the same notation to refer to node weights $c$,
edge weights $\omega$, node degrees $d(v)$, and the set of neighbors $\neighbors$.
In an undirected graph, an edge $(u,v) \in E$ implies an edge $(v,u) \in E$ and $\omega(u,v) = \omega(v,u)$.

A common way to represent an undirected hypergraph $H=(V,E,c,\omega)$ as an undirected graph is the \textit{bipartite} representation~\cite{HuMoerder85}.
In the \textit{bipartite} graph $G_*(V \dot\cup E, F)$, the vertices and nets of $H$ form the node set and for each net $e$ incident to a vertex $v$, we
add an edge $(e,v)$ to $G_*$~\cite{DBLP:conf/dac/SchweikertK72}. The edge set $F$ is thus defined as $F := \{(e,v)~|~e \in E, v \in e \}$.
An example of a hypergraph, along with its corresponding bipartite graph is shown in Figure~\ref{fig:hg_cg_bg}.

\paragraph{Partitions and Clusterings}
A \emph{$k$-way partition} of a hypergraph $H$ is a partition of its vertex set into $k$ \emph{blocks} $\Partition = \{V_1, \dots, V_k\}$
such that $\bigcup_{i=1}^k V_i = V$, $V_i \neq \emptyset $ for $1 \leq i \leq k$, and $V_i \cap V_j = \emptyset$ for $i \neq j$.
We call a $k$-way partition $\Partition$ \emph{$\varepsilon$-balanced} if each block $V_i \in \Partition$ satisfies the \emph{balance constraint}:
$c(V_i) \leq L_{\max} := (1+\varepsilon)\lceil \frac{c(V)}{k} \rceil$ for some parameter $\mathrm{\varepsilon}$.%
\footnote{The $\lceil\cdot\rceil$ in this definition ensures that there is always a feasible solution for inputs with unit vertex weights. For general weighted inputs there is no commonly accepted way how to deal with feasibility; see also \cite{HMS21,GHSSS21}.}
We call a block $V_i$ \emph{overloaded} if $c(V_i) > L_{\max}$ and \emph{underloaded} if $c(V_i) < L_{\max}$.

For each net $e$, $\conset(e) := \{V_i~|~ V_i \cap e \neq \emptyset\}$ denotes the \emph{connectivity set} of $e$.
The \emph{connectivity} $\con(e)$ of a net $e$ is the cardinality of its connectivity set, i.e.,  $\con(e) := |\conset(e)|$.
A net is called a \emph{cut net} if $\con(e) > 1$, otherwise  (i.e., if $|\mathrm{\lambda}(e)|=1$ ) it is called an \emph{internal} net.
A vertex $u$ that is incident to at least one cut net is called a  \emph{border vertex}.
The number of pins of a net $e$ in block $V_i$ is defined as  $\pinsinpart(e,V_i) := |\{V_i \cap e \}|$.
A block $V_i$ is \emph{adjacent} to a vertex $v \notin V_i$ if $\exists e \in  \incnets(v) : V_i \in \conset(e)$.
We use $\adjblocks(v)$ to denote the set of all blocks adjacent to $v$.
Given a $k$-way partition $\Partition$ of $H$, the \emph{quotient graph}
$Q := (\Partition, \{(V_i,V_j)~|~\exists e \in E : \{V_i,V_j\} \subseteq  \conset(e)\})$ contains an edge between each pair of adjacent blocks.

A \emph{clustering} $C = \{C_1, \dots, C_l\}$ of a hypergraph is a partition of its vertex set. In contrast to a $k$-way partition,
the number of clusters is not given in advance, and there is no balance constraint on the actual sizes of the clusters $C_i$.

\paragraph{Contractions and Uncontractions}
\emph{Contracting} a pair of vertices $(u, v)$ means merging $v$ into $u$. We refer to $u$ as the \emph{representative}
and to $v$ as the \emph{contraction partner}.
After contraction, the weight of $u$ becomes $c(u) := c(u) + c(v)$.  We connect $u$ to the former neighbors $\neighbors(v)$ of $v$, by replacing
$v$ with $u$ in all nets $e \in \incnets(v) \setminus \incnets(u)$. Furthermore, we remove $v$ from all nets $e \in \incnets(u) \cap \incnets(v)$.
If a contraction leads to parallel nets, we remove all but one from $H$. The weight of the remaining net $e$ is set to the sum of the weights of the nets parallel to $e$.
Single-vertex nets created by a contraction are removed from the hypergraph, since such nets can never become part of the cut set.
\emph{Uncontracting} a vertex $u$ reverses the contraction. The uncontracted vertex $v$ is put in the same
block as $u$ and the weight of $u$ is set back to $c(u) := c(u) - c(v)$.

\begin{figure}[t!]
  \centering
  \includegraphics[]{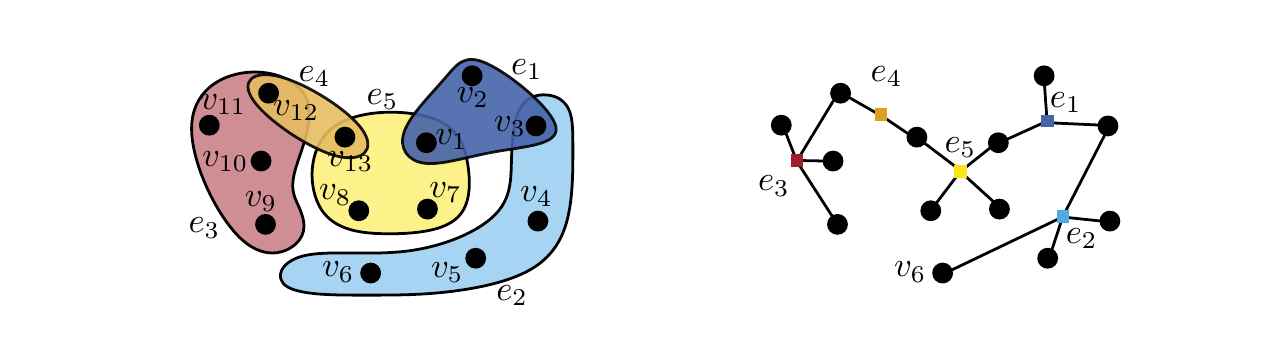}
  \caption[A hypergraph and its bipartite and clique-net representation.]{A hypergraph with $13$ vertices, $5$ nets, and $17$ pins (left)
    and the bipartite representation (right).}\label{fig:hg_cg_bg}
\end{figure}

\subsection{The \texorpdfstring{$k$}{k}-way  Hypergraph Partitioning Problem}\label{sec:hgp}
\paragraph{Problem Definition}
 The \emph{$k$-way hypergraph partitioning problem} is to find an $\varepsilon$-balanced $k$-way partition $\mathrm{\Pi}$ of a hypergraph $H=(V,E,c,\omega)$ that
 \emph{minimizes} an objective function over the cut nets for some value of~$\varepsilon$.
We consider the two most commonly used cost functions, namely the \emph{cut-net} metric $\ocut := \sum_{e \in E'} \omega(e)$ and the
\emph{connectivity} metric $\ocon := \sum_{e\in E'} (\lambda(e) -1)~\omega(e)$, where $E'$ is the \emph{cut-set} (i.e., the set of all cut nets)~\cite{donath1988logic,DBLP:journals/jpdc/DeveciKUC15}.
While the cut-net metric sums the weights of all nets that connect more than one block of the partition $\Partition$, the connectivity
metric additionally takes into account the actual number $\con$ of blocks connected by the cut nets.
Optimizing each of the objective functions is known to be NP-hard~\cite{Lengauer:1990}.
Note that for plain graphs both cost functions revert to edge-cut (i.e., the sum of the weights of those edges that have endpoints in different blocks) \cite{buluc2016recent}.
Furthermore, for $k=2$ (the \emph{bipartitioning} problem), $\ocut = \ocon$.

\paragraph{Recursive Bipartitioning and Direct $k$-way Partitioning}
In general, there are two approaches to computing a $k$-way partition of a hypergraph. If $k$ is a power of two, \emph{recursive bipartitioning} (RB)
algorithms obtain the final $k$-way partition by first computing a bipartition of the initial hypergraph, and then recursing on each of the two blocks.
Thus, it takes $\log k$ such phases%
\footnote{In this paper $\log x$ stands for $\log _2x$.} until the hypergraph is partitioned into $k$ blocks. If $k$ is not a power of two, the approach has to be adapted to
produce appropriately-sized blocks.
In \emph{direct $k$-way partitioning}, the hypergraph is directly partitioned into $k$ blocks, without the detour over the recursive $2$-way approach.

\section{High-Quality Hypergraph Partitioning}\label{s:algorithm}
This section presents the core data structures and algorithms of KaHyPar.
We start by describing our ``semi-dynamic'' hypergraph data structure in Section~\ref{sec:hypergraph_ds}. It is semi-dynamic in that
we are only concerned with efficient vertex and hyperedge \emph{deletions} and the reversal of these operations, and do not consider insertions of additional vertices or nets. In Section~\ref{sec:rb_vs_kway}, we briefly discuss our approach to computing $k$-way partitions via recursive bipartitioning and the peculiarities that need to be addressed for cut-net and connectivity optimization. Section~\ref{sec:preprocessing}
then presents the employed preprocessing techniques, namely the LSH-based sparsification algorithm and
an approach to infer information about the community structure of the hypergraph.
While the former is used to speed up the overall partitioning process, the information gathered by the latter is used to guide the coarsening process.
Afterwards, we address each of the three phases of the multi-level paradigm: Section~\ref{sec:coarsening} presents our coarsening algorithm, Section~\ref{sec:initial_partitioning} discusses our portfolio-based initial partitioning algorithm, and Section~\ref{sec:refinement} and Section~\ref{ssec:flows} give an overview on
our localized local search and flow-based refinement algorithms.
Lastly, in Section~\ref{subs:memetic_algo}, we integrate the entire framework with a genetic algorithm that is able to explore the global solution space extensively.

\subsection{The Hypergraph Data Structure}\label{sec:hypergraph_ds}

\begin{figure}[t]
  \centering
  \includegraphics[scale=0.9]{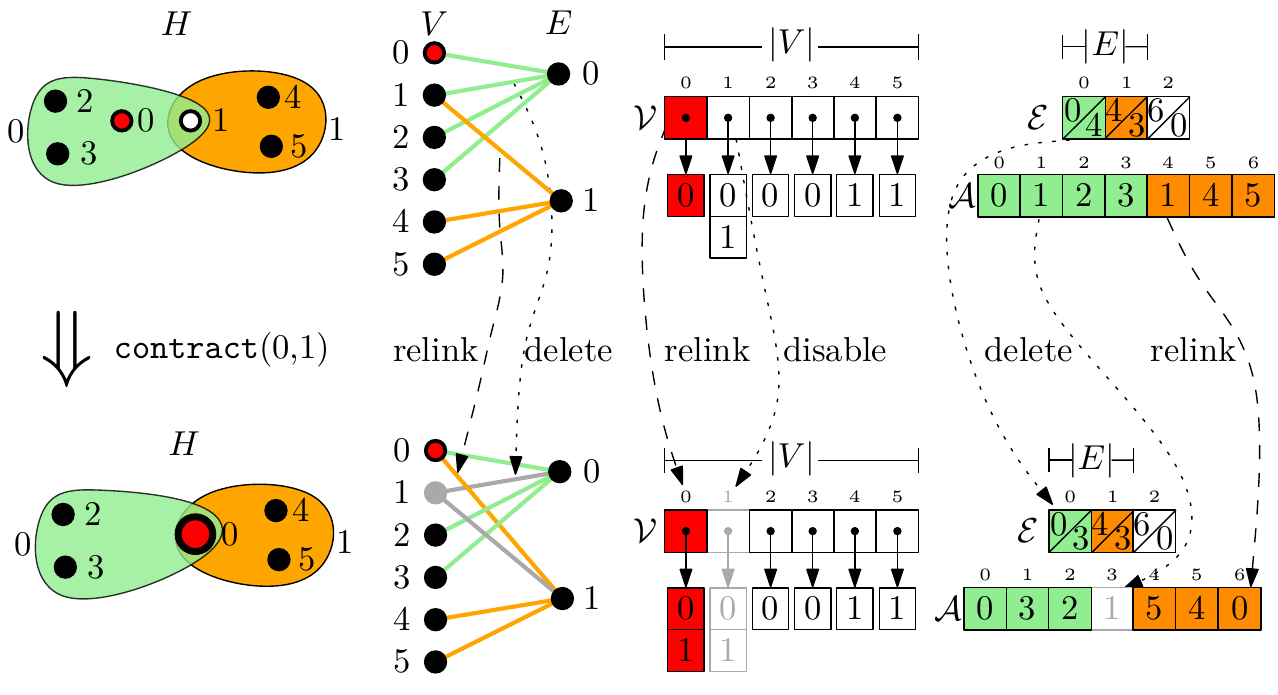}
  \caption[Example of a contraction operation.]{Example of a contraction operation. The hypergraph $H$ is depicted on the left, the corresponding bipartite graph representation is shown
    in the middle, and the adjacency data structure is shown on the right. The contraction leads to an edge deletion operation for net $0$ and
    a relink operation for net $1$. Element $\mathcal{E}[2]$ is a sentinel element used during uncontractions.}\label{fig:graphds}
\end{figure}

\paragraph{Conceptual Overview}
Conceptually, we represent the hypergraph $H$ as an undirected \emph{bipartite} graph $G_*=(V \dot \cup E, F)$. The vertices and nets
of $H$ form the vertex set. For each net $e$ incident to a vertex~$v$, we add an edge $(e,v)$ to the graph.
The edge set $F$ is thus defined as $F := \{(e,v)~|~e \in E \wedge v \in e \}$.
When contracting a vertex pair $(u,v)$, we mark the corresponding node $v$ as deleted. The edges $(v,e)$ incident to $v$ are treated as follows:
If $G_*$ already contains an edge $(u,e)$, then net $e$ contained both $u$ and $v$ before the contraction.
In this case, we simply delete the edge $(v,e)$ from $G_*$. Otherwise, net $e$ only contained $v$. We therefore have to
relink the edge $(v,e)$ to $u$.

\paragraph{Data Structure}
We use a combination of an adjacency list and a modified adjacency array to represent $G_*$.
An example is shown in Figure~\ref{fig:graphds}. The adjacency list is used to store
the incident nets of each vertex (i.e., the edges leaving nodes $v \in V$ in the bipartite graph representation), while the adjacency array stores the pins of each net (i.e.,
edges leaving nodes $v \in E$ in $G_*$). This representation is motivated by the observations that
vertex degrees can grow after a contraction, while nets can only shrink.
To index into the adjacency array $\mathcal{A}$, we use an offset array $\mathcal{E}$ that stores the starting positions of the entries in $\mathcal{A}$ ($\mathcal{E}[\cdot].f$)
and the size of each net $|e|$ ($\mathcal{E}[\cdot].s$).
Thus, pins of a net $e$ are accessible as $\mathcal{A}[\mathcal{E}[e].f], ...,  \mathcal{A}[\mathcal{E}[e].f+\mathcal{E}[e].s-1]$, while nets incident to a vertex $v$ are accessed using array $\mathcal{V}$ that stores a pointer for each vertex
to a vector containing the corresponding incident nets.

\paragraph{Contraction}
Contracting a vertex pair $(u,v) \in H$ works as follows: The weight of $v$ is added to the weight of $u$.
For each net $e \in \incnets(v)$ we then have to determine if the corresponding
edge $(v,e) \in G_*$ can simply be deleted or if a relink operation is necessary. This can be done with one scan over the pins of $e$. During
this scan, we swap $v$ with the last pin of $e$ located at position $\mathcal{A}[\mathcal{E}[e].f +\mathcal{E}[e].s-1]$, and additionally search for vertex $u$. If we find $u$, we can remove $v$ from $e$ by decrementing $E[e].s$.
If $u$ was not found, $(v,e)$ is relinked to $u$ by replacing $v$ with $u$ in the pin sequence of $e$ and by appending $e$ to the adjacency list of vertex $u$.
In order to be able to reverse contractions, we record each contracted vertex pair in a \emph{memento sequence $\mathcal{M}$}.

\paragraph{Uncontraction}

After 
 resetting the weight of the representative vertex $u$, we first mark all incident nets $\incnets(v)$ as relevant for the current uncontraction
using a bit vector $b$. We then iterate over all nets $e \in \incnets(u)$ of the representative $u$. If net $e$ is also incident
to the re-enabled vertex $v$ (i.e., $b[e] =~$true), it is necessary to revert either a delete or a relink operation. It is possible to distinguish between both cases by $\emph{peeking}$
one element past the slot of the last pin of $e$ ($E[e+1].f$ always exists, because we use a sentinel element at position $E[m]$). If the pin located at this position is $v$
and we are still in the pin range of net $e$ (i.e., the current size of $e$ is smaller than its original size), we have to revert a delete operation. Otherwise
a relink operation needs to be reverted. Deletions can be reversed by simply increasing $E[e].s$ for the corresponding nets.
To reverse a relink operation, we remove $e$ from the adjacency list of vertex $u$, and reset the pin slot of $e$ containing $u$ back to $v$.

\subsection{Computing \texorpdfstring{$k$}{k}-way Partitions via Recursive Bipartitioning}\label{sec:rb_vs_kway}

\paragraph{Motivation}
The question whether or not to prefer direct $k$-way partitioning over recursive bipartitioning is still
unresolved. Depending on the final application and algorithm, either one can be the method of choice. For example, recursive bipartitioning can be the method of choice if $k$ is not given in advance, but instead a bound on the block size is specified. In this case, recursive bipartitioning can stop as soon as the blocks reach the upper bound for the block size. Historically, recursive bipartitioning algorithms have been superior to $k$-way schemes for the hypergraph partitioning problem. Hence, our first partitioning algorithm~\cite{KaHyPar-R} used recursive bipartitioning to optimize the cut-net metric. Afterwards, we contributed a
direct $k$-way partitioning approach optimizing both the cut-net~\cite{KaHyPar-MF-JEA,KaHyPar-MF-SEA} and the connectivity metric~\cite{KaHyPar-MF-JEA,KaHyPar-MF-SEA,KaHyPar-CA,KaHyPar-K}. This is the first $k$-way approach that significantly outperforms recursive bipartitioning. However, recursive bipartitioning is still highly important as we use it within our initial partitioning algorithm.
In the following, we describe our approach for recursive bipartitioning.

\paragraph{Recursive Bipartitioning}
If $k$ is a power of two, the final $k$-way partition is obtained by first computing a bipartition of the
initial hypergraph and then recursing on each block. Hence, it takes $\log k$ such phases until the hypergraph is partitioned into $k$ blocks.
If $k$ is not a power of two, the approach has to be adapted to produce appropriately-sized partitions.
Our algorithm uses the following technique to compute a $k$-way partition via recursive bipartitioning for arbitrary values of $k$:
We compute a $2$-way partition of the hypergraph into blocks $A$ and $B$ such that $A$ has a maximum weight of $(1+\varepsilon') \lceil \lfloor k/2 \rfloor/k \cdot c(V) \rceil$ and $B$ has a maximum weight of $(1+\varepsilon') \lceil \lceil k/2 \rceil/k \cdot c(V) \rceil$,
where $\varepsilon'$ is a suitably adjusted imbalance parameter that ensures that the final $k$-way partition is $\varepsilon$-balanced.
Block $A$ is then partitioned recursively into $k' := \lfloor k/2 \rfloor$ blocks, while $B$ is partitioned into $k' := \lceil k/2 \rceil$ blocks.

\paragraph{Adaptive Imbalance}
If the initial imbalance parameter $\varepsilon$ was to be used in each bipartitioning step, the weight of the largest block $V_\text{max}$
could be larger than the maximum allowed block weight $L_{\max}$. Therefore it is necessary to restrict the allowed imbalance at each bipartition.
Our approach that balances flexibility over all levels is summarized in the following lemma:

\begin{lemma}[Adaptive Imbalance for Recursive Bipartitioning~\cite{KaHyPar-R}]
Let $H_0$ and $H_1$ be the hypergraphs induced by a bipartition $\Partition = \{V_0,V_1\}$ of an unweighted hypergraph $H=(V,E,c,\omega)$ for
which we want to compute an $\varepsilon$-balanced $k$-way partition. Using an adaptive imbalance parameter
\[
\varepsilon' := \left( \left(1+\varepsilon \right) \frac{k' \cdot c(V)}{k \cdot c(V_i)}\right)^{\frac{1}{\lceil \log k' \rceil}} -1
\]
to compute a $k'$-way partition (with $k'\geq 2$) of hypergraph $H_i$ with $i\in\{0,1\}$  via recursive bipartitioning ensures that the final $k$-way partition of $H$ is $\varepsilon$-balanced.
When computing the very first bipartition for a $k$-way partition, we set $H_0:=H$, $k':=k$ and therefore  $\varepsilon' := (1+\varepsilon)^{(1/\lceil \log k \rceil)} -1$.
\end{lemma}

\begin{proof}
We refer the reader to the thesis \cite{Schlag20} and the paper \cite{KaHyPar-R} for the proof of the lemma.
\end{proof}

 \paragraph{Cut-Net Splitting and Cut-Net Removal}
Depending on the objective function that we optimize, cut-nets need to be treated differently when recursing on the two  hypergraphs induced by a bipartition $\Partition = \{V_0,V_1\}$.
For cut-net optimization, we recurse on the \emph{section hypergraphs} $H_{\times V_0}$ and $H_{\times V_1}$ that omit the cut-nets completely, where
 $H_{\times V_i} :=(V_i, \{e \in E~|~e \subseteq V_i\})$~\cite{Berge:75,Berge:1985}.
Note that the section hypergraphs do not contain the cut-nets, because these nets will always be cut nets in the final $k$-way partition, and already contribute $\omega(e)$ to the total cut size~\cite{PaToHManual}.
This simultaneously reduces the number of nets as well as their average size
in each subhypergraph, without affecting the partitioning objective.
For connectivity optimization, however, all following bipartitions can further increase the connectivity $\lambda$ of the
cut nets. Therefore, it is necessary to recurse on the \emph{subhypergraphs} $H_{V_0}$ and $H_{V_1}$, in which each cut-net $e$
is \emph{split} into two nets $e_0 = e \cap V_0$ and $e_1 = e \cap V_1$. Single-pin nets can be discarded in this process, as they
cannot be cut in further bipartitioning steps.

\subsection{The Preprocessing Phase}\label{sec:preprocessing}
Before starting the $n$-level partitioning process, we apply a hypergraph sparsifier~\cite{KaHyPar-K} and infer information about the hypergraph's community structure~\cite{KaHyPar-CA}. Section~\ref{ssec:sparsification} describes the sparsification algorithm. Community
detection is discussed in Section~\ref{ssec:community_detection}. Whereas the former modifies the input hypergraph to speed up
the overall partitioning process, the latter gathers information that will be used to guide the coarsening process.

\subsubsection{Pin Sparsification via Locality-Sensitive Hashing}\label{ssec:sparsification}
\paragraph{Motivation}
Algorithms employed in each phase of the multi-level framework often perform computations on the vertices and their set of neighbors (e.g., to find the ``best''
contraction partner $u \in \neighbors(v)$ for vertex $v$ during coarsening). For a given vertex $v$, this requires iterating over the set of \emph{all} pins $p \in e$ of \emph{all} incident nets $e \in \incnets(v)$.
Especially for hypergraphs with many large nets, these calculations can therefore have a significant impact on the overall running time of the respective
algorithm.  To alleviate this impact, we employ a \emph{pin} sparsifier as a preprocessing technique that identifies (and contracts) vertices with similar neighborhoods and thus
reduces the average hyperedge size~\cite{KaHyPar-K}.

\paragraph{Central Idea}
We consider two vertices $u$ and $v$ to be similar, if they \emph{share} many nets, i.e., if their sets of incident nets $\incnets(u)$ and $\incnets(v)$ have a relatively large intersection.
Similarity is measured using the Jaccard coefficient $J(A, B) =  |A \cap B|/|A \cup B |$, for finite sets $A$ and $B$. Let
$J(u, v)$ denote the Jaccard coefficient of the incident nets of two vertices $u$ and $v$, i.e., $J(u, v) =  |\incnets(u) \cap \incnets(v)|/|\incnets(u) \cup \incnets(v)|$.
Then, the corresponding \emph{distance} metric is $D(u,v) = 1-J(u, v)$~\cite{DBLP:conf/stoc/Charikar02}.
Since calculating these distances/similarities for every pair of vertices would lead to a quadratic-time algorithm, we instead use
the \emph{locality-sensitive} hashing (LSH) technique~\cite{Indyk:1998:ANN:276698.276876,Gionis:1999:SSH:645925.671516} to identify sets of similar vertices that are ``close''
to each other with respect to $D(\cdot,\cdot)$. Similar vertices  are then contracted to reduce the number of pins in the hypergraph.

\paragraph{Locality-Sensitive Hashing (LSH)}
The key idea of the LSH approach is to hash elements in such a way that the probability for ``close'' elements to have equal hash values is high, while the probability for ``distant'' elements
to have equal hash values is low. For the Jaccard distance $D(\cdot,\cdot)$, the following family of hash functions (called \emph{min-hash}) is known to be locality-sensitive:
$\mathcal{H} = \{h_\sigma(X) = \min\{\sigma(x) \mid x \in X \mid \sigma \in \Sigma\}\}$, where $X \subseteq U$ is a finite set of elements from a finite universe $U$, and $\sigma$ is a random permutation from the set $\Sigma$ of all random permutations
of $U$~\cite{Broder:1997:RCD:829502.830043,DBLP:journals/cn/BroderGMZ97}.
It can be proven that $Pr[h_\sigma(A) = h_\sigma(B)] = J(A,B)$~\cite{DBLP:journals/cacm/AndoniI08,DBLP:conf/www/LiK10}, i.e.,
the \emph{larger} the distance, the \emph{smaller} the collision probability~\cite{Gionis:1999:SSH:645925.671516}; see also \cite{DBLP:books/cu/LeskovecRU14}. For increased efficiency, we replace random permutations $\sigma$ by hash functions of the form
$h(x) = ax+b\bmod p$ \cite{Broder:2000:MIP:348360.348768}.

\paragraph{High-Level Algorithm Outline}
The algorithm works in multiple passes. In the beginning, all vertices are marked as unclustered. Each pass then starts by identifying buckets of similar vertices using min-hash fingerprints. The fingerprint of a vertex $v$ is defined as $g_i(v) = (h_{1}(v), h_{2}(v), \cdots, h_{i}(v))$, where each hash function $h_j$ is chosen uniformly at random from $\mathcal{H}$. We consider two fingerprints to be equal if and only if all $i$ hash values are equal. Hence, only vertices with the \emph{same} fingerprint will be put in the same bucket.
The size of the fingerprint (i.e., the number $i$ of min-hashes) affects the probability that vertices are put into the same cluster~\cite{DBLP:books/cu/LeskovecRU14}. By increasing the number of hashes, we decrease  the probability that ``distant'' vertices have the same fingerprint.
However, at the same time, increasing the size of the fingerprint also decreases the probability of ``close'' vertices ending up in the same cluster.
To avoid this problem, we calculate multiple fingerprints for each vertex. Since the distance between vertices varies in different parts of a hypergraph, we choose both the size of the fingerprint and the number of fingerprints adaptively.

After computing the buckets, each yet unclustered vertex is then clustered with similar vertices from it's corresponding bucket as long as the size of the resulting cluster is less than $c_\text{max}$. If the size is at least $c_\text{min}$, all vertices of the newly formed cluster become inactive and do not participate in the next pass.
By bounding cluster sizes from below by $c_\text{min}$ and from above by $c_\text{max}$, we enforce the formation of reasonably balanced clusters in order to allow the partitioning algorithm to compute feasible solutions of high quality.
The clustering algorithms stops as soon as the number of resulting clusters is less than $n/2$ or the maximum number of passes is reached.
Each cluster is then contracted to a single vertex. We implement the sparsifier in such a way that the total running time is linear in the total number of pins.

\subsubsection{Detecting Community Structure To Improve Coarsening}\label{ssec:community_detection}
\paragraph{Motivation}
The goal of the coarsening phase is to create successively smaller but \emph{structurally similar} approximations of the input hypergraph in which both the exposed
hyperedge weight as well as the sizes of the hyperedges are successively reduced. This is commonly done by using rating functions to identify
and contract highly connected vertices, and by allowing the formation of vertex \emph{clusters} instead of enforcing maximal matchings,
since this can destroy some naturally existing clustering structures in the hypergraph~\cite{Karypis2003}.
However, even when allowing the formation of vertex clusters, several situations can arise in which the naturally existing structure is obscured. Examples include tie-breaking decisions if multiple neighbors of a vertex have the same rating score, or
preventing certain contractions by enforcing an upper-bound on the vertex weights to ensure that the distribution of vertex weights
does not become too imbalanced at the coarsest level (which would limit the number of feasible initial partitions satisfying the balance constraint).

Situations like these arise because all coarsening algorithms are guided by local, greedy decisions based on rating functions that solely consider
the weights and sizes of nets connecting candidate vertices and therefore lack a global view of the clustering problem.
We therefore use an approach which incorporates \emph{global} information about the \emph{community structure} into the coarsening process~\cite{KaHyPar-CA}.

\paragraph{Community Detection via Modularity Maximization}
We perform community detection on hypergraphs by translating this problem into the problem
of modularity maximization in graphs.
Community detection tries to extract an underlying structure from a graph by dividing its nodes
into disjoint subgraphs (communities) such that connections are dense \textit{within} subgraphs but
sparse between them~\cite{Fortunato201075,Schaeffer07graph}. Different quality functions are used to judge
the goodness of a division into communities. Among those, the most popular quality function is the \textit{modularity} measure \cite{Newman04}. It
compares the observed fraction of edges within a community with the expected fraction of edges if edges were placed using a random edge distribution that
preserves the degree distribution of the graph~\cite{Fortunato20161}. More formally, given a graph $G$ and disjoint communities $C=\{C_1,\dots,C_x\}$, modularity is defined as:
\begin{equation}\label{eq:modularity}
  Q:= \frac{1}{2m} \sum_{ij}\left[A_{ij} - \frac{k_ik_j}{2m}\right]\delta(C_i,C_j),
\end{equation}
where $A_{ij}$  is the entry of the adjacency matrix $A$ representing edge $(i,j)$, $m=\frac{1}{2}\sum_{ij}A_{ij}$ is the number of edges in the graph,
$k_i$ is the degree of node $i$, $C_i$ is the community of vertex $i$, and $\delta$ is the Kronecker delta. Note that this can be generalized to
weighted graphs: $A_{ij}$ represents the weight of edge $(i,j)$, $k_i=\sum_jA_{ij}$ is the weighted degree of node $i$ and $m=\frac{1}{2}\sum_{ij}A_{ij}$ is the sum of all edge weights~\cite{WeightedMod}.
Modularity optimization
is known to be NP-hard~\cite{modularity.np}, but several efficient heuristics exist. In KaHyPar, we use the Louvain algorithm of \citet{Louvain},
which has low computational complexity and is thus suitable for large graphs~\cite{Fortunato201075,Fortunato09}.

\paragraph{Community-aware Coarsening Framework}
Our framework consists of two phases. First, we use the Louvain algorithm to partition the vertices of the hypergraph
into a set $C = \{C_1,\dots,C_x\}$ of internally densely and externally sparsely connected communities. The actual number of
communities $|C|$ is determined by the community detection algorithm. Then, the hypergraph coarsening algorithm described in Section~\ref{sec:coarsening}
 is applied on each community $C_i$ independently. This can be accomplished by modifying the algorithm to only contract vertices
within the \emph{same} community by restricting potential contraction partners of a given a vertex $u \in C_i$ to $\neighbors(u) \cap C_i$.
By preventing inter-community contractions, the coarsening algorithm maintains the structural similarity
discovered by the community detection algorithm, while still allowing local, intra-community decisions to be based
on HGP-specific rating functions.

\paragraph{Hypergraph Clustering by Modularity Optimization of the Bipartite Representation}
We apply graph modularity optimization to the representation of the hypergraph by the bipartite graph $G_*$ already described in the introduction. We can then infer a clustering of the hypergraph from the clustering of the vertex side of $G_*$.
In order to make this work, we encode information about the hypergraph structure into the edge weights of $G_*$ (see Figure~\ref{fig:bipartite}). When the \emph{edge density}
$\delta=m/n$ is sufficiently large ($\delta\geq 0.75$ in our implementation), we use constant edge weights.
For less dense inputs, we use an edge weight $\omega(v, e) := d(v)/|e|$.
By weighting the graph edge inversely proportional to the size of the net,
smaller nets get a higher influence on the community structure than larger nets. If many small nets are contained within a community, the coarsening
algorithm can successively reduce their size and eventually remove them from the hypergraph.
Furthermore, this ensures that large nets do not dominate the community structure by attracting too many vertices.
By also incorporating the vertex degree, we strengthen the connection between nets and high-degree vertices to facilitate
the formation of communities around high-degree vertices in the hypergraph.

\begin{figure}[t!]
  \centering
  \includegraphics[]{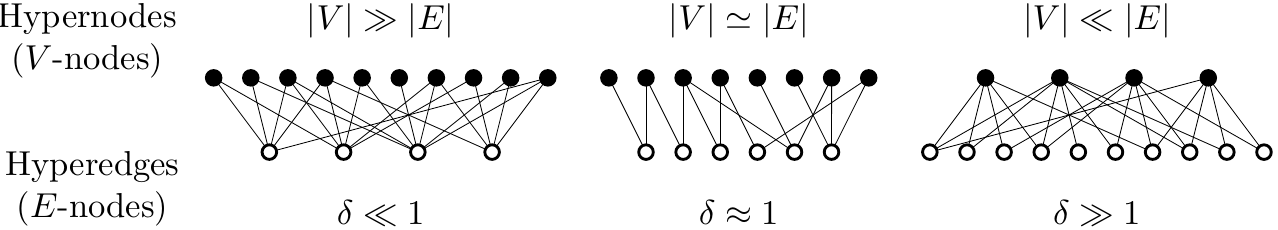}
  \caption[Star-expansion of hypergraphs of varying edge density $\delta$. ]{Bipartite graph-based representations of hypergraphs of varying edge density $\delta$. For hypergraphs with $\delta \ll 1$, the bipartite
    graph consists of many $V$-nodes with low average degree and fewer $E$-nodes with high average degree (left).
    If $\delta \approx 1$, the number of $V$- and $E$-nodes and their average degrees are roughly equal (middle). Hypergraphs with high ratio $\delta \gg 1$
    lead to bipartite representations with fewer $V$-nodes with high average degree and many $E$-nodes with low
    average degree (right) (adapted from \cite{KaHyPar-CA}).}\label{fig:bipartite}
\end{figure}

\subsection{The Coarsening Phase}\label{sec:coarsening}
\paragraph{Motivation}
Multi-level coarsening algorithms either compute matchings~\cite{AlpertHK98,KarypisAKS99,DBLP:journals/siamrev/VastenhouwB05,DBLP:conf/ipps/DevineBHBC06}
or clusterings~\cite{DHauBo97,DBLP:journals/vlsi/KarypisK00,DBLP:journals/tpds/CatalyurekA99,DBLP:conf/ispdc/TrifunovicK04,DBLP:journals/heuristics/MeyerhenkeS016} on each level of the coarsening hierarchy
using different rating functions to determine the vertices to be matched or clustered together.
The contracted vertices then form the vertex set of the coarser hypergraph on the next level.
In contrast, $n$-level partitioning algorithms like the graph partitioner KaSPar~\cite{DBLP:conf/esa/OsipovS10} create a hierarchy of (nearly) $n$ levels
by removing only a \emph{single} vertex between two levels, which obviates the need for matching or clustering algorithms in the coarsening phase.
KaSPar uses a priority queue to determine which vertex pair to contract next. After each contraction, the priority of each neighbor of the
contracted vertex needs to be updated in order to keep the priorities consistent. For hypergraphs, this constitutes a
severe performance bottleneck as even a \emph{single} large hyperedge can significantly increase the size of the neighborhood.
In this section, we present a simple coarsening algorithm that avoids this bottleneck but retains the solution quality of the KaSPar approach~\cite{KaHyPar-K}.

\paragraph{Rating Function}
Our algorithm adopts the \emph{heavy-edge} rating function \cite{KarypisAKS99,DBLP:journals/jpdc/TrifunovicK08,PaToHManual}, which prefers vertex
pairs  $(u,v)$ that have a large number of heavy nets with small size in common:
\begin{equation}\label{eq:our_rating_funct}
\text{r}(u,v) :=\sum \limits_{e \in \{\incnets(v) \cap \incnets(u)\}}  \frac{\omega(e)}{|e| - 1}.
\end{equation}

\paragraph{Detecting Single-Vertex Nets and Parallel Nets}
The contraction of a vertex pair $(u,v)$ can lead to parallel nets (i.e., nets that contain exactly the same vertices) and single-vertex nets in $\incnets(u)$. We continuously detect and remove these nets from the hypergraph. Parallel  nets are replaced by a single net whose weight is the sum of the weights of the parallel nets.
 Single-vertex nets are easily identified, because $|e|=1$.
Parallel net detection works by computing a fingerprint $f_e := \sum_{v \in e} v^2$ where squaring acts as a simple and fast hash function%
\footnote{Using a slightly more expensive hash function, this approach has concrete performance guarantees \cite{DBLP:conf/ipps/Hubschle-Schneider18}.} and where the commutativity of the sum operation ensures that the order of pins does not matter \cite{DBLP:journals/siamsc/HendricksonR98}.
In order to efficiently support the $n$-level approach, $f_e$ is updated incrementally: If net $e$ contained $v$ before the contraction, $v^2$ is subtracted from $f_e$. If $e$ contained $v$ but not $u$, relinking implies that $u^2$ has to be added to $f_e$.

\paragraph{Algorithm Outline}
In general, our algorithm is very similar to the First Choice (FC) algorithm employed in hMETIS-K~\cite{DBLP:journals/vlsi/KarypisK00}. It works in multiple passes.
At the beginning of each pass, we create a random permutation of the current vertex set. For each vertex $u$, we then determine its eligible contraction partner $v \in \neighbors(u)$ with highest rating $r(u,v)$
and immediately contract $(u,v)$ \emph{on the fly}. Thus, $v$ will be removed from the hypergraph and will not be visited in this or any future passes over the vertex set.
To avoid imbalanced inputs for the initial partitioning phase, a neighboring vertex $v$ is only considered as an eligible contraction partner
 if $c(u) + c(v) \leq \kappa$,
where $\kappa := \lceil \frac{c(V)}{t\cdot k} \rceil$ is
the maximum allowed vertex weight. Here, parameter $t$ is used to control the size of the coarsest hypergraph.
If multiple neighbors have the same rating score, we favor a vertex that has not yet taken part in any contractions during the pass to break ties.
To speed up the coarsening process in the presence of large nets, we do not evaluate the rating function for nets larger than $\iota$ vertices.
A pass ends as soon as every vertex in the random permutation was
considered either as representative or as contraction partner. Then, a new pass is started by creating a new random permutation of the remaining vertices. The coarsening process is stopped as soon as the number of vertices drops below $t \cdot k$ or no eligible vertex is left.
This algorithm has been shown to be significantly faster than an engineered, hypergraph-specific version of the KaSPar approach~\cite{KaHyPar-R} while yielding comparable solution quality~\cite{KaHyPar-K}.

\paragraph{A Different View on Coarsening}
The key difference to FC and related coarsening algorithms lies in the way contractions are handled. While traditional algorithms \emph{first} compute a
matching/clustering on each level and \emph{then} use it to create a coarse hypergraph for the next level, we rate \emph{and} contract one vertex at the same time, i.e.,
after finding the contraction partner $v \in \neighbors(u)$ for a vertex $u$, we immediately contract $(u,v)$.
Thus, while in FC clustering decisions are made for all vertices of the current level \emph{at once}, and more importantly \emph{independently} of one another,
our $n$-level coarsening algorithm adaptively adjusts \emph{every} contraction decision to the current structure of the hypergraph induced by all previous contractions.

\subsection{The Initial Partitioning Phase}\label{sec:initial_partitioning}
In the direct $k$-way setting, KaHyPar computes an initial partition of the coarsest hypergraph via $n$-level recursive bipartitioning
as described in Section~\ref{sec:rb_vs_kway}. To obtain an initial bipartition, we use a portfolio approach composed of nine different
initial bipartitioning algorithms. The portfolio approach increases diversification and produces
better results than single initial partitioning algorithms alone~\cite{HeuerBA}. Each algorithm runs
$20$ times using different random seeds and the quality of each computed partition is improved using our $2$-way FM local search algorithm (see Section~\ref{subs:2way_fm}).
The partition with the best solution quality or, in case of ties, with the lowest imbalance, is
used as initial solution and projected back to the original hypergraph. In case all partitions are imbalanced, we choose the partition with
smallest imbalance.

In the following, we give a brief overview of the bipartitioning algorithms employed in our portfolio approach and
refer to the corresponding bachelor thesis~\cite{HeuerBA} for a more detailed description and evaluation.

\paragraph{Random \& BFS-based Partitioning}
Random partitioning~\cite{PaToHManual, KarypisAKS99} randomly assigns vertices to one of the two blocks, provided that the assignment does not violate the
balance constraint. In case of a violation, the vertex in question is assigned to the opposite block. If both assignments would lead
to overloaded blocks, the vertex is randomly assigned to one of the blocks.
Breadth-First-Search (BFS) partitioning~\cite{PaToHManual,DBLP:journals/siamsc/KarypisK98,KarypisAKS99} starts with a randomly chosen vertex and performs a BFS traversal
of the hypergraph until the weight of all discovered vertices would exceed the balance constraint. The vertices visited during the traversal constitute the first block $V_0$,
all remaining vertices constitute the second~block~$V_1$.

\paragraph{Greedy Hypergraph Growing (GHG)}
Furthermore, we use different variations of the GHG algorithm proposed by \citet{DBLP:journals/tpds/CatalyurekA99}.
Unlike the original algorithm, which grows a cluster around a randomly selected seed vertex, our versions first compute two pseudo-peripheral seed vertices~\cite{george1981computer} as follows:
Starting from a random vertex, we perform a BFS. The last vertex visited serves as the starting vertex
for the next BFS. This vertex and the last vertex visited by the second BFS are supposed to be ''far'' away from each other.
Therefore, one is used as the seed vertex for block $V_0$, the other for block $V_1$.  For each block, we maintain a priority queue (PQ) that stores the neighboring vertices
of the growing cluster according to a score function. The algorithm then iteratively selects the vertex with the highest score from one of the PQs, moves the vertex
to the corresponding block, and then updates the scores of neighboring vertices.
We use the FM gain (see Eq.~\ref{eq:gain})
and the \emph{max-net} gain definition (which is also used in PaToH~\cite{PaToHManual}) as score functions.
The max-net gain counts the weights of all nets connected to the target block
(i.e., the gain of assigning vertex $v$  to block $V_i$ is defined as $g_\text{max-net}(v) := \sum_{e \in E'} \omega(e)$, where $E' := \{e \in \incnets(v)~|~\pinsinpart(e,V_i) > 0\}$).
Our GHG variants also differ in the way the clusters are grown.
The \emph{global} strategy always moves the vertex with the highest score of both PQs to the corresponding block, whereas the
\emph{sequential} approach only grows block $V_0$ and assigns the remaining vertices to block $V_1$
(also implemented in PaToH~\cite{DBLP:journals/tpds/CatalyurekA99} and Metis~\cite{DBLP:journals/siamsc/KarypisK98}).
The \emph{round-robin} technique grows both blocks simultaneously.
In total, the initial partitioning portfolio therefore contains six different initial partitioning algorithms based on GHG.

\paragraph{Size-Constrained Label Propagation}
The last algorithm is based on the adaptation of size-constrained label propagation (SCLaP)~\cite{DBLP:journals/heuristics/MeyerhenkeS016} to HGP local search.
The SCLaP-based refinement algorithm was initially proposed in the master thesis of Vitali Henne~\cite{HenneMA}.
Each vertex has a label representing its block.  Initially all labels are empty, i.e., all vertices are unassigned.
The algorithm starts by searching two pseudo-peripheral vertices via BFS as described above. One vertex and $\tau$ of its neighbors then get label $V_0$, while the other
vertex and $\tau$ of its neighbors get label $V_1$.  The algorithm then works in rounds until it has converged, i.e., no empty labels remain.
In each round, the vertices are visited in random order and each vertex $u$ is assigned the label of the neighbor $v \in \neighbors(u)$ that results
in the highest FM gain, provided that the resulting cluster does not become overloaded. Ties are broken randomly.
Once the algorithm has converged,  vertices with the same label then become a block of the bipartition. The tuning parameter $\tau$ is used to prevent labels from disappearing over the course of the algorithm,
and, based on experimental results~\cite{HeuerBA}, is set to $\tau=5$  in our implementation.

\subsection{Localized \texorpdfstring{$2$}{2}-way and \texorpdfstring{$k$}{k}-way FM Local Search}\label{sec:refinement}
\paragraph{Overview}
We now turn to our local improvement algorithms. Both $2$-way and $k$-way local search follow the FM paradigm~\cite{FiducciaM82} and are further inspired by the algorithm used in
KaSPar~\cite{DBLP:conf/esa/OsipovS10}. A key difference to the traditional FM algorithm is the way a local search pass is started:
Instead of initializing the algorithm with all vertices or all border vertices, we perform a \emph{highly localized}  search starting only with the representative and the just uncontracted vertex.
The search then gradually expands around this vertex pair by successively considering neighboring vertices.
Our $2$-way local search algorithm optimizing the cut-net metric $\ocut$ is described in Section~\ref{subs:2way_fm}. It is also used
to implicitly optimize the connectivity metric $\ocon$ when KaHyPar is configured to use recursive bipartitioning.
In Section~\ref{subs:kway_fm}, we then describe our $k$-way local search algorithm. Unlike in the case of $2$-way partitioning,
objective-specific gain computations and delta-gain updates are necessary to permit the algorithm to optimize both objectives.
Traditional multi-level FM implementations as well as KaSPar \emph{always} compute the gain of \emph{each} vertex from scratch at each level of the hierarchy.
During an FM pass, these values are then either kept up-to-date by delta-gain updates~\cite{DBLP:journals/tpds/CatalyurekA99,DPapa2007} or recomputed whenever necessary~\cite{DBLP:phd/dnb/Schulz13a}.
Since our algorithms start around only two vertices, many gain values would never be used during a local search pass.
In Section~\ref{subs:gain_caching}, we therefore propose a \emph{gain caching technique} that ensures that the gain of a vertex move is calculated at most \emph{once} during \emph{all}
local searches along the $n$-level hierarchy.
Since local search is done after each uncontraction, it is necessary to limit the number of vertex moves in
each pass, because otherwise the $n$-level approach could lead to a quadratic
number of local search steps in total. In Section~\ref{subs:stopping_rules}, we therefore present
the \emph{adaptive stopping rule} that terminates the iterative improvement process before all vertices
have been moved.

\subsubsection{2-way Localized FM Refinement}\label{subs:2way_fm}
\paragraph{Algorithm Outline}
We use two PQs to maintain the possible moves for all vertices -- one for each block.
At the beginning of a local search pass, both queues are empty and disabled. A disabled PQ will not be considered when searching for the next move with the highest gain.
Initially, all vertices are labeled \emph{inactive}.
Only inactive vertices are allowed to become \emph{active}.
To start the local search phase after each uncontraction, we activate the representative and the just uncontracted vertex if
they are border vertices. Otherwise, no local search phase is started.
\emph{Activating} a vertex $v$ currently assigned to block $V_i$ means that we calculate the FM gain $g_j(v)$ for moving $v$ to another block $V_j \in \adjblocks(v) \setminus \{V_i\}$ and
insert $v$ into the corresponding queue $P_j$ using $g_j(v)$ as key. The FM gain of a vertex $v \in V_i$  is defined as
\begin{equation} \label{eq:gain}
g_{j}(v) := \omega(\{ e~|~e \in \incnets(v):~\pinsinpart(e, V_j) = |e| - 1\}) - \omega(\{ e~|~e \in \incnets(v):~\pinsinpart(e, V_i) = |e|\}).
\end{equation}

After insertion, PQs corresponding to \emph{underloaded} blocks become enabled. Since a move to an overloaded block will
never be feasible, any queue corresponding to an overloaded block is left disabled.
The algorithm then  repeatedly queries only the \emph{non-empty, enabled} queues to find the move with the highest gain $g_{j}(v)$,
breaking ties arbitrarily. Vertex $v$ is then moved to block $V_j$ and labeled as \emph{marked}.
We then update all neighbors $\neighbors(v)$ of $v$ as follows: All previously inactive neighbors are activated as described above.
All active neighbors that have become internal are labeled inactive and the corresponding moves are deleted from the PQs.
Finally, we perform \emph{delta-gain updates} for all moves of the remaining active border vertices in $\neighbors(v)$:
If the move changed the gain contribution of a net $e \in \incnets(v)$, we account for that change by
incrementing/decrementing the gains of the corresponding moves by $\omega(e)$ using the delta-gain-update algorithm of Papa and Markov~\cite{DPapa2007}.
Once all neighbors are updated, local search continues until either no non-empty, enabled PQ remains or the stopping rule mandates the termination of the current pass.
After local search is stopped, we reverse all moves until we arrive at the lowest cut state reached during the search that fulfills the balance constraint.
All vertices become inactive again and the algorithm is then repeated until no further improvement is achieved.

\paragraph{Locked Nets}
To further decrease the running time, we exclude nets from gain updates that cannot be removed from the cut in the current local search pass.
A net is \emph{locked} in the bipartition once it has at least one marked pin in each of the two blocks \cite{Krishnamurthy84}.
In this case, it is not possible to remove such a net from the cut by moving any of the remaining movable pins to another block.
Thus, it is not necessary to perform any further delta-gain updates for locked nets, since their contribution to the gain values of their pins does not change any more.

\subsubsection{\texorpdfstring{$k$}{k}-way Localized FM Refinement}\label{subs:kway_fm}
There is a large design space for $k$-way refinement algorithms. On one extreme, there is the single-level $k$-LA$_\ell$-FM algorithm
of Sanchis~\cite{Sanchis89,Sanchis93} which maintains $k(k-1)$ priority queues (one for each possible move direction). On the
other extreme, we have the $k$-way local search technique of KaFFPa~\cite{kaffpa}  that uses a single priority queue
which only stores the highest-gain move for every vertex. In between, there is the rotary KLFM algorithm of
Chan et al.~\cite{MLHtr,Zien:1997:HSP:266388.266523} and the K-PM approach of Cong and Lim~\cite{DBLP:conf/iccad/CongL98}. Rotary KLFM
uses  $2(k-1)$ PQs and in each round only allows moves between a target block $V_i$ and all other blocks $\Partition \setminus \{V_i\}$,
while K-PM only uses two PQs and iteratively improves the $k$-way partition by moving vertices between all $k(k-1)/2$ pairs of blocks.
Our algorithm is based on the refinement scheme used in KaSPar (i.e., we use $k$ priority queues and each PQ stores the vertex moves to that particular block), because $k$-LA$_\ell$-FM is too expensive for large $k$~\cite{Sanchis89,Sanchis93},
both rotary KLFM and K-PM only have a restricted view on the $k$-way partition, and the KaFFPa approach makes it necessary
to recompute gains after each move in order to identify those with highest gain.

\paragraph{Differences to the $2$-way Algorithm}\label{kfm-outline}
In general, the $k$-way refinement algorithm follows the same outline
as the $2$-way algorithm described in the previous section. We
therefore focus on the differences to the $2$-way algorithm before
describing the gain computation and delta-gain update techniques for
$k$-way connectivity and cut-net optimization.  We only consider
moving a vertex $v \in V_i$ to \emph{adjacent} blocks
$\adjblocks(v)\setminus\{V_i\}$ rather than calculating and
maintaining gains for moves to \emph{all} $k$ blocks.  This
simultaneously reduces the memory requirements (space bounded by the
number of border vertices) and restricts the
search space of the algorithm to moves that are more likely to improve
the solution. A PQ $P_i$ for
block $V_i$ is enabled if $P_i$ is not empty and $V_i$ is not \emph{overloaded}.
The $2$-way FM algorithm implicitly forces
unbalanced solutions to become balanced whereas the $k$-way algorithm only
maintains feasibility. This is due to the fact that we disable PQs of overloaded blocks.

\paragraph{Connectivity Metric: Gain Computation \& Delta-Gain Updates}
When activating a vertex $v \in V_i$, we calculate the \emph{gain} $g_{j}(v)$ for moving $v$ to all adjacent blocks $V_j \in \adjblocks(v) \setminus \{V_i\}$, and insert  $v$ into the corresponding queues $P_j$ using $g_j(v)$ as key.
For connectivity optimization, the  gain $g_j(v)$ is defined as
\begin{equation}\label{eq:kway_con_gain}
g_j(v) := \omega(\{ e~|~e \in \incnets(v):~\pinsinpart(e, V_i) = 1\}) - \omega(\{ e~|~e \in \incnets(v)~\pinsinpart(e, V_j) = 0\}).
\end{equation}

After moving a vertex $v$, we perform \emph{delta-gain updates} for all moves of the remaining active border vertices in $\neighbors(v)$.
If the move changed the gain contribution of a net $e \in \incnets(v)$, we account for that change by increasing/decreasing the gains of the corresponding moves by $\omega(e)$.
Moving a vertex $v$ can furthermore change the connectivity $\con(e)$ of a net $e \in \incnets(v)$, which in turn can affect the set of adjacent blocks $B(\cdot)$ for each neighbor in $\neighbors(v)$.
The delta-gain-update algorithm takes these changes into account by inserting moves to newly adjacent blocks into the PQs and removing moves to blocks that
are not adjacent anymore. For more details about the different delta-gain update cases and how handle it in case we optimize the cut metric, we refer the reader
to the corresponding papers \cite{KaHyPar-K, Schlag20}.

\paragraph{Excluding Nets from Delta-Gain Updates}
To further reduce the running time of the delta-gain algorithms, we exclude nets from the update procedure if their contribution to the gain values of their pins cannot change.
For connectivity optimization,
the key observation is that after moving a vertex $v$ to a block $\topart$, this block remains connected to all nets $e \in \incnets(v)$ during this local search pass,
because $v$ is not allowed to be moved again. In this case we say that block $\topart \in \conset(e)$ has become \emph{unremovable} for net $e$.
We exclude nets $e \in \incnets(v)$ from delta-gain updates after moving a vertex
$v$ from $\frompart$ to $\topart$ if both blocks  $\{\frompart, \topart\} \in \conset(e)$ are marked as unremovable:

\subsubsection{Caching Gain Values}\label{subs:gain_caching}
\paragraph{Gain Cache for $2$-way Refinement}
We use a cache array $\rho$ where $\rho[v]$ denotes the cache entry for vertex $v$.
After initial partitioning, the gain cache is empty. If a vertex becomes activated during a local search pass, we check whether or not its gain is already cached.
If it is cached, the cached value is used for activation. Otherwise, we calculate the gain according to Eq.~\ref{eq:gain},
insert it into the cache and activate the vertex. After moving a vertex $v$ with gain $g_{j}(v)$ to block $V_j$, its cache value is set to $\rho[v] := - g_{j}(v)$.
The delta-gain updates of its neighbors $\neighbors(v)$ are then also applied to the corresponding cache entries.
Thus, the gain cache always reflects the current state of the hypergraph.
Since our algorithm performs a  rollback operation at the end of a local search pass that undoes vertex moves, we also have to undo delta-gain updates applied on the cache.
This can be done by additionally maintaining a \emph{rollback} \emph{delta-gain} \emph{cache} that stores the negated delta-gain updates for each vertex. During rollback,
this delta cache is then used to restore the gain cache to a valid state.

Each time a local search is started with an uncontracted vertex pair $(u,v)$, we have to account for the fact that the uncontraction potentially affected~$\rho[u]$.
A simple variant of the caching algorithm just recalculates $\rho[u]$. Since $v$ did not exist on previous levels of the hierarchy,
its gain must also be computed from scratch. For $2$-way refinement, we instead use a more sophisticated variant that is able to update $\rho[u]$ based
on information gathered during the uncontraction and that further infers $\rho[v]$ from $\rho[u]$.
For more details on how to update cache entries $\rho[u]$ and $\rho[v]$ after uncontraction, we refer the
reader to \cite{KaHyPar-R}.

\paragraph{Gain Data Structure for $k$-way Refinement}
In order to generalize the $2$-way gain cache to $k$-way partitioning, a redesign of the data structure is necessary.
Since in the $2$-way setting there is only one possible move for each vertex (i.e., moving it to the other block of the bipartition), a simple array is enough to store the gain values.
When performing $k$-way local search, each vertex can potentially be moved to $k-1$ different blocks. We therefore use a modified version of a folklore data structure to store
sparse sets (see, e.g., Ref.~\cite{Briggs:1993}). The $k$-way gain cache uses for each vertex such
a sparse set that stores for each adjacent block the corresponding gain values. Adding, removing or
updating a cache entry can be done in constant time and the data structure uses $\Oh{k}$ space for each
vertex.

\paragraph{Gain Cache for $k$-way Refinement}
Let $\rho_{v}[j]$ denote the cache entry for vertex $v$ and adjacent block $V_j \in \adjblocks(v)$.
After initial partitioning, we initialize the gain cache with all possible moves of all vertices of the coarsest hypergraph.
Each time a local search is started with an uncontracted vertex pair $(u,v)$, we invalidate and recompute their corresponding cache entries.
This is necessary since $v$ did not exist on previous levels of the hierarchy and since the uncontraction potentially affected both $\adjblocks(u)$ and
the corresponding gain values.\footnote{We also tried a version that -- similar to $2$-way gain caching -- updates the cache entries of $u$ based on the information
  gathered during uncontraction and then infers the cache entries of $v$ from those of $u$. However, recomputation turned out to be faster, since
updating and inferring the cache values is significantly more complicated.}
If a vertex becomes activated during a local search pass, the cached gain values are used for activation.
After moving a vertex $v$ with gain $g_{j}(v)$ from block $\frompart$ to block $\topart$, its cache value is updated as follows:
First, we remove the entry of $\rho_{v}[\text{\upshape to}]$, since $v \in \topart$ after the move and we only cache gain values
for moves to adjacent blocks $\adjblocks(v) \setminus \{\topart\}$. If $v$ remains connected to $\frompart$, we set $\rho_{v}[\text{\upshape from}] := - \rho_{v}[\text{\upshape to}]$.
For more details on how to update the remaining cache entries $\rho_{v}[i]$ with
$V_i \in \adjblocks(v) \setminus \{\frompart, \topart\}$
see \cite{KaHyPar-K}.

After updating the cache entries of the moved vertex, delta-gain updates of the neighbors $\neighbors(v)$ are then also applied to the corresponding cache entries.
Thus the gain cache always reflects the current state of the hypergraph. Similarly to the $2$-way case, we additionally maintain a \emph{rollback} \emph{delta} \emph{cache}
that stores the negated delta-gain updates for each vertex as well as the corresponding add/remove operation for $\adjblocks(\cdot)$ in order to
be able to restore the gain cache to a valid state during rollback.

\subsubsection{Restricting the Search Space via Adaptive Stopping}\label{subs:stopping_rules}
Unlike traditional multi-level partitioning algorithms with only few levels, which can afford to spend linear time in refinement heuristics at each hierarchy level,
the number of local search steps needs to be limited in the $n$-level setting. Otherwise the $n$-level approach could lead
to $\Oh{n^2}$ local search steps in total, if refinement is executed after every uncontraction. We therefore use a stopping rule that terminates
a local search pass well before every vertex is moved once.

The adaptive stopping criterion is a slightly modified version of the stopping rule proposed by Osipov~\etal~\cite{DBLP:conf/esa/OsipovS10} for $n$-level graph partitioning.
This approach approximates the local search as a random walk, i.e., gain values in each step are assumed to be identically distributed and independent random
variables. Based on the average gain $\mu$ since the last improvement and the variance $\sigma^2$  observed during the current local search it is shown that it is unlikely to still obtain an improvement after $p>\sigma^2/4\mu^2$ steps of the local search.
We integrate a slightly refined version of this adaptive stopping criterion into our algorithm:
On each level, local search performs at least $\log n$ steps after an improvement is found and continues as long as $\mu > 0$.
If $\mu$ is still 0 after $\log n$ steps, local search is stopped. This prevents the algorithm from getting stuck with zero-gain moves, which is likely for hypergraphs
that contain many large nets. Otherwise (i.e., if $\mu \neq 0$) we evaluate the equation and act accordingly.

\subsection{Flow-Based Refinement}\label{ssec:flows}
\paragraph{Motivation}
Move-based local search algorithms are prone to get stuck in local optima when used directly on the input hypergraph~\cite{DBLP:journals/vlsi/KarypisK00}.
The multi-level paradigm helps to some extent, since it allows a more global view of the problem at the coarse levels.
However, the solution quality still degrades for hypergraphs with large hyperedges.
Since large hyperedges are likely to have many vertices in multiple blocks, it is difficult to find impactful moves~\cite{DBLP:journals/siamsc/UcarA04}.
Thus, the gain of moving a single vertex to another block is likely to be \emph{zero}~\cite{DBLP:conf/sat/MannP14}.

While finding \emph{balanced} minimum cuts is NP-hard, a minimum cut separating two vertices can be found in polynomial time using the max-flow min-cut theorem~\cite{DBLP:conf/stoc/GoldbergT86} and maximum flow algorithms.
Flow algorithms find an optimal min-cut and do not suffer the drawbacks of move-based approaches.
However, they were long overlooked as heuristics for balanced partitioning due to their high complexity~\cite{GraphEdgeReduction,552086}.

\paragraph{Overview}
\citet{kaffpa} present a max-flow-based refinement algorithm for graph partitioning which is integrated into the multi-level partitioner KaFFPa and computes high-quality solutions.
\citet{552086} as well as~\citet{flowcutter} propose a flow-based bipartitioning algorithm (FlowCutter) which solves incremental max-flow instances to converge towards a balanced bipartition.
\citet{hfc_esa} employ FlowCutter as a refinement algorithm for hypergraph bipartitions.

We integrate FlowCutter refinement into our $n$-level framework.
Since it is not feasible to run this refinement on every level of the $n$-level hierarchy, we run it after uncontracting $i = 2^j$ vertices for increasing $j \in \mathbb{N}_+ $.
For $k$-way partitions we apply the refinement to block pairs using the active block scheduling approach of~\citet{kaffpa}.

We briefly outline how to refine a bipartition with FlowCutter, before filling in the details on the different steps.
Given a bipartition $(V_0, V_1)$ to refine, we decide which vertices can be moved -- denote these by $M$.
We contract the vertices $V_0 \setminus M$ to a vertex $s$, and $V_1 \setminus M$ to $t$.
Subsequently, we run FlowCutter on this hypergraph with $s$ as source and $t$ as sink to obtain an improved bipartition.
Since FlowCutter will separate $s$ and $t$, only the vertices in $M$ can change their block.

We describe the method used to find $M$ in Section~\ref{sec:flows:extraction}, then the FlowCutter algorithm in Section~\ref{sec:flows:hyperflowcutter}, and finally how to compute maximum flows on hypergraphs in Section~\ref{sec:flows:on-hypergraphs}.

\subsubsection{Identifying Vertices to Move}\label{sec:flows:extraction}

We use the same method as \citet{kaffpa}, adjusted for hypergraphs.
Vertices close to the cut between $V_0$ and $V_1$ are good candidates for $M$, since they are the most likely to improve the solution when moved.
Therefore, two BFSs are performed. One that only visits $V_0$ and one that only visits $V_1$.
They are initialized with the border vertices of $V_0$, $V_1$, respectively.
Every visited vertex is added to $M$, and the BFSs run until some weight constraint on $M \cap V_i$ would be violated.
\citet{kaffpa} use $(1 + \alpha \cdot \epsilon) \frac{c(V)}{k} - c(V_{1 - i})$ to bound the weight of $M \cap V_i$, where $\alpha \geq 1$ is an input parameter.
For $\alpha = 1$ this guarantees that the minimum cut is balanced.
While KaFFPa tries different values for $\alpha$, we only need one, since FlowCutter guarantees balanced partitions.
We use $\alpha = 16$ since this is the maximum value used in our previous work~\cite{KaHyPar-MF-JEA} that implemented the KaFFPa framework on hypergraphs.

The contracted hypergraph contains all hyperedges with pins in $M$, but we remove hyperedges with only one remaining pin.
Our framework is able to optimize the \emph{connectivity} as well as the \emph{cut-net} metric.
For the cut-net metric, we additionally omit hyperedges with pins in blocks other than $V_0, V_1$ since they cannot be removed from the cut.

\subsubsection{FlowCutter}\label{sec:flows:hyperflowcutter}

We take as input a (hyper)graph as well as an initial source $s$ and target vertex $t$.
First, we compute a maximum flow, which yields a source-side and a target-side minimum cut.
If either of these is balanced, we are done.
Otherwise, we transform all vertices on the smaller side to a source -- if the source-side is smaller -- or a target otherwise.
Further, we add one more vertex to the same side -- the \emph{piercing vertex} -- so that the next steps find different cuts.
Subsequently, we augment the previous flow to a maximum flow of the new network.
These steps are repeated until balance is achieved or the cut of the input bipartition is exceeded.
Piercing vertices are chosen incident to the cut of the smaller side.
If possible, we choose vertices that do not create an augmenting path as this improves balance without increasing the cut size.
Ties are broken randomly.

Once we find the first balanced bipartition, we try to improve its balance by keeping the algorithm running as long as the piercing vertices do not create augmenting paths.
More balanced solutions tend to give the FM local search more leeway for improvement.
Since this is very fast, we perform several repetitions, each starting at the first balanced bipartition.
This approach is similar to the \emph{most-balanced-minimum-cut} heuristic of KaFFPa~\cite{kaffpa}.

\paragraph{Maximum Flows on Hypergraphs}\label{sec:flows:on-hypergraphs}

\begin{figure}[t!]
	\centering
	\includegraphics[width=0.7\textwidth]{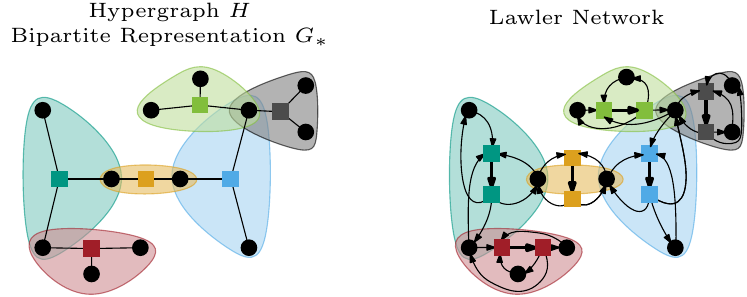}
	\caption[Comparison of different hypergraph flow networks.]
	{Unweighted Hypergraph $H$ with overlayed bipartite representation $G_*$ and illustration of the corresponding Lawler network.
		Thin, directed edges have infinite capacity, thick edges have unit capacity.}\label{fig:flow_networks}
\end{figure}

To compute min cuts we use maximum flows on hypergraphs, which are a straight-forward generalization from graphs.
In the \emph{maximum flow problem on hypergraphs}, we are given a hypergraph $H = (V,E)$, a capacity function $\mathpzc{c} : E \to \mathbb{N}_0$ on the hyperedges and two distinguished vertices $s$ and $t$.
An $s$-$t$-flow $f$ is a function, which assigns each pin the \emph{amount of flow} $f(u,e)$ it sends into the corresponding hyperedge.
A negative value means that the pin receives flow.
The flow on a hyperedge is $f(e) := \sum_{u \in e} \max(0, f(u,e))$.

A flow $f$ is feasible if it satisfies the following constraints.
The \emph{flow conservation constraints} $\forall u \in V \setminus \{s,t\}: \sum_{e \in I(u)} f(u,e) = 0$ and $\forall e \in E: f(e) = \sum_{u \in e} \max(0, -f(u,e))$ ensure that flow is created only at $s$ and drained only at $t$.
The \emph{capacity constraint} $\forall e \in E: f(e) \leq \mathpzc{c}(e)$ restricts the amount of flow on a hyperedge.

A flow is maximum if its value $|f| := \sum_{e \in I(s)} f(u,e)$ is maximum among all feasible flows.
The max-flow min-cut theorem states that the value $|f|$ of a maximum flow $f$ is equal to the capacity of a min-cut
separating $\mathpzc{s}$ and $\mathpzc{t}$~\cite{ford1956maximal}.

For each hypergraph, we can define an equivalent graph-based flow network, called the \emph{Lawler network} $L$~\cite{DBLP:journals/networks/Lawler73}.
$L$ contains two vertices $e_\text{in}, e_\text{out}$ for each hyperedge $e \in E$, and one vertex $v$ for each hypervertex $v \in V$.
For each hyperedge $e \in E$, it contains the edge $(e_\text{in}, e_\text{out})$ with capacity $\omega(e)$, as well as the edges $(v, e_\text{in})$ and $(e_\text{out}, v)$ with infinite capacity, for each $v \in e$.
See Figure~\ref{fig:flow_networks} for an illustration of $L$.
While we can run any flow algorithm on $L$, this is slow in practice.
Therefore, we never explicitly construct $L$ but adapt our flow algorithm so that it runs directly on the hypergraph -- now viewed as an implicit representation of $L$. We use an adaptation of Dinic' algorithm~\cite{10021311931}.

The algorithm consists of two alternating phases that are repeated until the flow cannot be augmented: assigning hop-distance labels to vertices by performing a BFS on residual edges (capacity not exceeded), and using DFS to find shortest edge-disjoint augmenting paths with distance labels increasing by one along the path.
By using two \emph{distance labels} for the hyperedges (representing $e_\text{in}, e_\text{out}$) in addition to those for the vertices, we can implement the BFS and DFS on the hypergraph.
The major performance benefit compared to the Lawler network comes from an optimization where we eliminate iterations over pins.
Consider a saturated hyperedge $e$, i.e., $f(e) = c(e)$ and a pin $u \in e$ from which we try to push flow via $e$.
There are two cases: either $u$ receives flow from $e$, in which case we can push that amount back and send it to any other pin, or we can only push flow to pins of $e$ that send flow into $e$.
In the latter case, we only need to scan those pins.
We achieve this by maintaining a partitioning of the pins of each net into three different sets: \emph{sends flow}, \emph{receives flow}, and has \emph{no flow}.

We refer the reader to the conference paper~\cite{whfc_sea} for more information, e.g., on how to perform augmentations, data structures and optimization techniques.

\newif\ifEnableExtendMemetic
\EnableExtendMemeticfalse

\subsection{Memetic Strategies}\label{subs:memetic_algo}
\paragraph{Motivation}
The story of high-quality hypergraph partitioning as told so far is very much about improving on simple local search. Preprocessing, the multi-level paradigm, and flows on pairs of blocks all mitigate this problem
but all lack a really global view. Repeated executions can help, but still only scratch the surface of the huge space of possible partitionings.
Hence, more sophisticated metaheuristics are needed.

\ifEnableExtendMemetic
Several genetic and memetic hypergraph partitioning algorithms have already been proposed in the literature~\cite{areibi2000integrated,areibi2004effective,ArmstrongGAD10,DBLP:conf/dac/BuiM94,DBLP:conf/gecco/KimKM04}.
However, \emph{none of them} is considered to be truly competitive with state-of-the-art tools~\cite{Cohoon2003}. We believe that this is due to the fact that all of them employ
 \emph{flat} (i.e., non multi-level) partitioning algorithms to drive the exploitation of the local solution space. \fi{}
We therefore integrate the $n$-level hypergraph partitioning framework presented in the previous sections with a genetic algorithm and thus develop
 the first multi-level memetic algorithm for the hypergraph partitioning problem~\cite{KaHyPar-E}.
\ifEnableExtendMemetic
 \paragraph{References}
 This chapter is based on a technical report~\cite{KaHyPar-E-TR} and a conference paper jointly published with Robin Andre and Christian Schulz~\cite{KaHyPar-E}. The paper was mainly written by the author of
 this dissertation, with editing by Christian Schulz. Large parts of this chapter were copied verbatim from the paper.
 The initial implementation of the  evolutionary framework was done by Robin Andre as part of his bachelor thesis~\cite{AndreBA},
 which was supervised by us. This implementation was then improved and integrated into the KaHyPar framework by the author of this dissertation.
 The experimental evaluation presented in Section~\ref{sec:evo_eval} contains some experimental results of Ref.~\cite{KaHyPar-MF-JEA}.\fi{}

\paragraph{Overview}
We start by explaining the components of our memetic $n$-level hypergraph partitioning algorithm.
Given a hypergraph $H$ and a time limit $t$, the algorithm starts by creating an initial
population $\mathcal{P}$ of \emph{individuals}, which correspond to $\varepsilon$-balanced $k$-way partitions of $H$.
The population size $|\mathcal{P}|$ is determined dynamically by first measuring the time $t_\text{I}$ spent to create one individual.
Then, $\mathcal{P}$ is chosen such that the time to create $|\mathcal{P}|$ individuals is a certain percentage $\eta$ of the total running time $t$:
$|\mathcal{P}| := \max(3,\min(50,\eta\cdot(t/t_I)))$, where $\eta$ is a tuning parameter.
We set the \emph{fitness} of an individual to the connectivity $\ocon$ or the cut-net metric $\ocut$ of its partition $\Partition$.
The initial population is evolved over several generational cycles using the \emph{steady-state} paradigm~\cite{EvoComp}, i.e., only \emph{one} offspring is created per generation using a recombination or mutation operation.
In order to sufficiently explore the global search space and to prevent premature convergence,
 we use mutation operators based on V-cycles~\cite{WalshawVcycle} that
exploit knowledge of the problem domain.
Furthermore, we propose a replacement strategy which considers
fitness \emph{and} similarity to determine the individual to be evicted from the population.

\subsubsection{Recombination Operators}\label{sec:combine}
We now describe the recombination operators. We generalize the recombine operator framework of KaFFPaE~\cite{kaffpaE} from
graphs to hypergraphs. Hence, the two-point recombine operator described in this section ensures that the fitness of the offspring is
\emph{at least as good as the best of both parents}. The edge frequency-based multi-point recombination operator described afterwards
gives up this property, but still generates good offspring.

\ifEnableExtendMemetic
\begin{figure}[t!]
  \centering
  \includegraphics[]{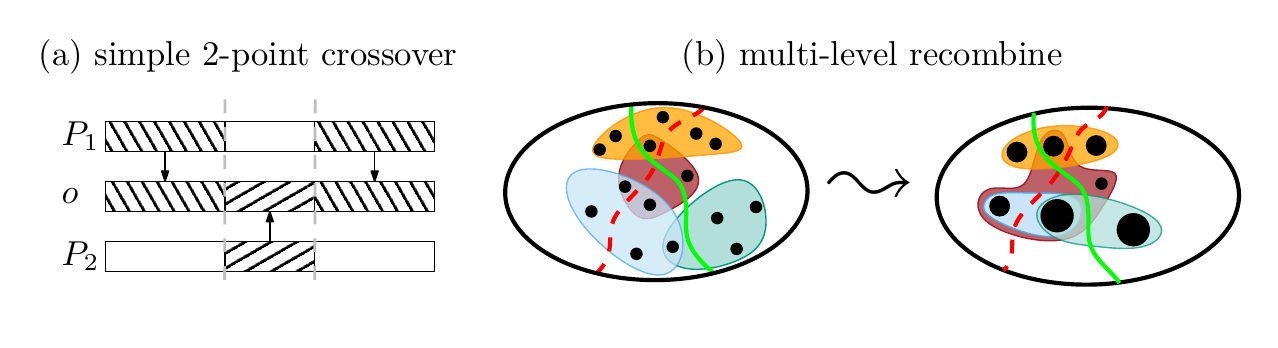}
  \caption[Traditional recombination and modified multi-level coarsening.]{(a) Traditional, problem agnostic crossover operation to combine parent partitions $P_1$ and $P_2$ to offspring $o$.
    (b) Recombination using modified multi-level coarsening to combine two partitions (dashed red line and solid green line).
    Each cut net $e$ remains in the coarse hypergraph and maintains its connectivity $\lambda(e)$ regarding both partitions (source: \cite{KaHyPar-E}).}\label{fig:combine}
\end{figure}
\fi{}

\paragraph{Two-Point Recombination}
The operator starts with selecting parents for recombination using binary tournament selection (without replacement)~\cite{BlickleT96}.
Two individuals $I_1$ and $I_2$ are chosen uniformly at random from $\mathcal{P}$ and
the individual with better fitness  becomes the first parent~$P_1$. We repeat this process to get the second parent $P_2$.
Both parents are then used as input of a modified $n$-level partitioning scheme as follows:
During coarsening, two vertices $u$ and $v$ are only allowed to be contracted if \emph{both parents agree
on the block assignment of both vertices}, i.e., if $b_1[u] = b_1[v] \wedge b_2[u] = b_2[v]$.
This restriction ensures that cut nets $e$ remain in the coarsened hypergraph
\emph{and} maintain their connectivity $\lambda(e)$ regarding \emph{both} parent partitions.
This  allows us to use the partition of the \emph{better} parent as initial partition of the offspring.
We alter the stopping criterion of the coarsening phase such that it stops when no more contractions are possible.
For more details, see the conference version of this algorithm~\cite{KaHyPar-E}.

\ifEnableExtendMemetic
The high-quality solution of the coarsest hypergraph contains two different classes of vertices: Those for which both parent partitions agree on a block assignment
and those for which they do not (see Figure~\ref{fig:combine} (b) for an example).
During the uncoarsening phase, refinement algorithms can then use this initial partitioning to (i) exchange good parts of the solution on the coarse
levels by moving few vertices and (ii) to find the best block assignment for those vertices, for which the parent partitions disagreed.
Since KaHyPar's refinement algorithms guarantee non-decreasing solution quality, the fitness of offspring solutions generated using
this kind of recombination is always \emph{at least as good as the better of both parents}.
\fi{}

\paragraph{Edge-Frequency Multi-Recombination}
The operator described previously is restricted to recombine $p=2$ partitions to improved offspring of non-decreasing quality.
Sanders and Schulz~\cite{kaffpaE} specifically restrict their operators to this case, arguing that in the course of the algorithm
a series of two-point recombine operations to some extend emulates a multi-point recombination.
Here, we present a multi-point recombine operation.
We perform a detailed experimental evaluation to test this hypothesis in \cite{KaHyPar-E} and show that adding type of recombine operations can yield better solutions in practice than just using simple two-way recombination operations.
The recombine operator uses the concept of (hyper)edge frequency~\cite{Wichlund98} to pass information about the cut nets of the $t$ \emph{best individuals} in the population on to new offspring.
The frequency $f(e)$ of a net $e$ hereby refers to the number of times it appears in the cut in the $t$ best solutions: $f(e) := |\{I \in t~|~\lambda(e) > 1\}|$. We use $t = \lceil \sqrt{|\mathcal{P}|} \rceil$,
which is a common value in evolutionary algorithms~\cite{delling2010graph}.
The multi-recombine operator then uses this information to create a \emph{new} individual in the following way. The
coarsening algorithm is modified to prefer to contract vertex pairs $(u,v)$ which share a large number of small, low-frequency nets. This is achieved
by replacing the standard heavy-edge rating function of KaHyPar with the rating function~\cite{Wichlund98} shown in Eq.~\ref{ef}:

\begin{equation}\label{ef}
r(u,v) := \frac{1}{c(v) \cdot c(u)}~\sum \limits_{e \in \{\mathrm{I}(v) \cap \mathrm{I}(u)\}}  \frac{\exp(-\zeta f(e))}{|e|}.
\end{equation}

This rating function disfavors the contraction of vertex pairs incident to cut nets with high frequency, because these nets
are likely to appear in the cut of high-quality solutions. The tuning parameter $\zeta$ is used as a damping factor.
After coarsening stops,  KaHyPar's initial partitioning algorithms are used to compute an initial partition of the coarsest hypergraph, which is then refined
during the uncoarsening and local search phase.

\subsubsection{Mutation Operators} We define two mutation operators based on V-cycles~\cite{WalshawVcycle}.
Both operators are applied to a random individual $I$ of the population.
This approach has been applied successfully as mutation operator
in evolutionary graph partitioning~\cite{kaffpaE}, therefore we also adopt it for HGP.
We reuse an already computed partition as input for the $n$-level approach, by restricting contractions to vertex pairs belonging to the same block.
This way, the existing partition can be carried over to the coarsest hypergraph with equivalent solution quality.
By distinguishing two possibilities for initial partitioning, we define two different
mutation operators: The first one uses the current partition of the individual as initial partition of the coarsest hypergraph and guarantees non-decreasing
solution quality. The second one employs KaHyPar's portfolio of initial partitioning algorithms to compute a \emph{new} solution for the coarsest hypergraph.
During uncoarsening, local search algorithms improve the solution quality and thereby further mutate the individual.
Since the second operator computes a new initial partition which might be different from the original partition of $I$, the fitness of offspring generated by this operator can be worse than the fitness of $I$.

\subsubsection{Replacement Strategy}\label{sec:replace}
All recombination and
mutation operators create one new offspring $o$.  In order to keep the
population diverse, we evict the individual \emph{least different} from $o$ among all individuals
whose fitness is equal to or worse than $o$.
More precisely, for each individual we compute the multi-set $D :=  \{(e, \lambda(e)-1) : e \in E \} $, where $\lambda(e)$ is the number of blocks connected by $e$. The difference between two individuals with multisets $D_1$ and $D_2$, respectively is then the size of the symmetric difference of $D_1$ and $D_2$.
Recall that the multiplicity of an element $e$ in the symmetric difference of $D_1$ and $D_2$ is defined as $\max(m_1,m_2)-\min(m_1,m_2)$ when $m_1$ is the multiplicity of $e$ in $D_1$ and $m_2$ is the multiplicity of $e$ in $D_2$.

\section{Experimental Evaluation}\label{s:experiments}
Our implementations of the proposed algorithms form the core of the $n$-level hypergraph partitioning framework \emph{KaHyPar}
(\textbf{Ka}rlsruhe \textbf{Hy}pergraph \textbf{Par}titioning)\footnote{KaHyPar is available from \url{https://github.com/kahypar/kahypar}}.
KaHyPar supports direct $k$-way partitioning and
recursive bipartitioning, and can optimize the cut-net and connectivity metric. In this evaluation, we focus
on optimizing the connectivity metric using the direct $k$-way approach, since it performs better than recursive bipartitioning and the overall results for cut-net and connectivity optimization are similar~\cite{Schlag20}.

Our experimental evaluation is organized as follows: After discussing the framework configuration in Section~\ref{subs:algorithm_configuration}, we first evaluate the effects of KaHyPar's different algorithmic components on the running
time and solution quality in Section~\ref{sec:kahypar_config}. Then, we compare KaHyPar with current state-of-the-art hypergraph partitioners using a
large real-world benchmark set composed of $488$ hypergraphs and discuss their time-quality tradeoffs in Section~\ref{sec:partitioner_comparison}.
We conclude the evaluation with two case studies on edge partitioning and traditional graph partitioning in Sections~\ref{sec:graph_edge_partitioning} and~\ref{sec:graph_vertex_partitioning}, which demonstrate the robustness of the techniques presented in this paper.

\paragraph{Instances.}
Our main benchmark set (referred to as set A) contains $488$ hypergraphs that are
derived from four different sources:
the ISPD98 VLSI Circuit Benchmark Suite~\cite{ISPD98},
the DAC 2012 Routability-Driven Placement Contest~\cite{DAC2012},
the SuiteSparse Matrix Collection~\cite{FloridaSPM},
and the 2014 SAT Competition~\cite{SAT14Competition}.
We translate sparse matrices to hypergraphs using the row-net model~\cite{DBLP:journals/tpds/CatalyurekA99} and SAT
instances to three different hypergraph representations: \emph{literal}, \emph{primal}, and \emph{dual}~\cite{DPapa2007, DBLP:conf/sat/MannP14}
(for more details see also~\cite{KaHyPar-CA}). All hypergraphs have unit vertex and net weights.
Additionally, we use two different subsets of set A. The first contains $164$ hypergraphs (referred to as set B) and was assembled
such that it reflects the qualitative results of the partitioners on the entire benchmark set~\cite{KaHyPar-R}. The second contains
$100$ hypergraphs (referred to as set C) for which each tested partitioner can compute a partition in under eight hours~\cite{KaHyPar-K}.
The benchmark sets are publicly available~\cite{SchlagBenchmarkSets}
and their detailed statistics can be found in Figure~\ref{fig:benchmark_set} in Appendix~\ref{app:benchmarks}.

\paragraph{System and Methodology.}
The code is written in \Cpp{} and compiled using
\gpp{9.2} with flags \texttt{-O3 -mtune=native -march=native}.
All experiments are performed on a cluster consisting of machines with Intel Xeon Gold 6230 processors 
 running at $2.1$ GHz with $96$GB RAM.
Unless mentioned otherwise, all hypergraphs are partitioned with an allowed imbalance of
$\varepsilon = 0.03$ into $k \in \{2, 4, 8, 16, 32, 64, 128\}$ blocks.
For each value of $k$, a $k$-way partition is considered to be \emph{one} test instance, resulting
in a total of $\numprint{3416}$, $\numprint{1148}$, and $\numprint{700}$ instances for benchmark sets A, B, and C, respectively.
We partition each instance ten times using different random seeds and
aggregate running times and solution quality using the arithmetic mean over all seeds.
To further aggregate over multiple instances, we use the geometric mean for absolute running times.
Runs with imbalanced partitions are not excluded from aggregated running times.
For runs that exceeded a predefined time limit, we use the time limit itself in the aggregates.
Only if all runs of an algorithm exceeded the time limit for a specific instance, we mark it with \ClockLogo~in the plots.
Similarly, if all runs of an algorithm  produced imbalanced partitions on an instance, we mark it with \ding{55}.

\paragraph{Performance Profiles.}
To compare the solution quality of different algorithms, we use \emph{performance profiles}~\cite{DBLP:journals/mp/DolanM02}.
Let $\mathcal{A}$ be the set of all algorithms we want to compare, $\mathcal{I}$ the set of instances, and $q_{A}(I)$ the average quality of algorithm
$A \in \mathcal{A}$ on instance $I \in \mathcal{I}$.
For each algorithm $A$, we plot the fraction of instances ($y$-axis) for which $q_A(I) \leq \tau \cdot \min_{A' \in \mathcal{A}}q_{A'}(I)$, where $\tau$ is on the $x$-axis.
Achieving higher fractions at lower $\tau$-values is considered better.
For $\tau = 1$, the $y$-value indicates the percentage of instances for which an algorithm performs best.
Note that these plots relate the quality of an algorithm to the best solution and thus do not permit a full ranking of three or more algorithms.

\paragraph{Wilcoxon Signed Rank Test.} When performance profiles yield inconclusive results, we additionally perform Wilcoxon signed ranked tests~\cite{wilcoxon} to decide whether or not the differences in solution quality
are statistically significant. At a $1\%$ significance level ($p \le 0.01$), a $Z$-score with $|Z| \ge 2.576$ is considered
significant~\cite[p.~180]{WilcoxonZValues}.

\subsection{Framework Configuration}\label{subs:algorithm_configuration}
\paragraph{Preprocessing.}
Pin sparsification is enabled for hypergraphs with median net size $|\tilde{e}| \ge 28$. The minimum cluster size for
sparsification $c_\text{min}$ is set to two and the maximum cluster size $c_\text{max}$ is set to ten.
For community detection, the edge weighting scheme is chosen dynamically at runtime depending on the edge density $\delta$ of the hypergraph. If $\delta \geq 0.75$, constant edge weights
are used, otherwise we use and edge weigth of $\omega(v,e)=d(v)/|e|$. Furthermore, we restrict the Louvain algorithm to perform at most $100$ iterations on each level and stop the first phase of the
algorithm if the improvement in modularity is below $0.0001$.

\paragraph{Coarsening.}
The threshold parameter $\iota$ for evaluating the rating function is set to $\iota = 1000$. The coarsening process is stopped once the number of vertices drops below $160\cdot k$ (i.e., $t=160$) or no eligible vertex is left.

\paragraph{Initial Partitioning.}
Intitial partitions are computed by executing our $n$-level recursive bipartitioning algorithm on the coarsest hypergraph.
The algorithm uses cut-net splitting when configured to optimize the connectivity metric $\ocon$ and cut-net removal
for $\ocut$-optimization (see Section~\ref{sec:rb_vs_kway}).
The coarsening process continues until the number of vertices drops below $300$ (i.e., $k=2$ and $t=150$). Initial bipartitions
are computed using the portfolio approach. Bipartitions are refined by the $2$-way localized local search algorithm using the simple stopping rule.
At each level, refinement stops after $i=50$ moves neither improved the solution quality nor the current imbalance.

\paragraph{Refinement.}
After computing the initial $k$-way partition via recursive bipartitioning, it is further refined using the localized $k$-way local search
algorithm described in Section~\ref{subs:kway_fm} and the flow-based refinement algorithm described in Section~\ref{ssec:flows}. For the former, the search space is restricted using the adaptive stopping rule.
For the FlowCutter refinement described in Section~\ref{ssec:flows}, we set the maximum weight for vertices to include from block $V_i$ in the flow model to $(1 + \alpha \cdot \epsilon) \frac{c(V)}{k} - c(V_{1 - i})$, with $\alpha = 16$, as in Ref.~\cite{KaHyPar-MF-JEA, whfc_sea}.
The most-balanced-minimum-cut heuristic is repeated 7 times.

\subsection{Evaluating Algorithmic Components}
\label{sec:kahypar_config}
\paragraph{Motivation}
KaHyPar contains three optional algorithmic components that are not part of the core $n$-level algorithm: pin-sparsification, community-aware coarsening, and flow-based refinement.
In order to evaluate the importance of these algorithmic components, we start with the KaHyPar configuration described in Section~\ref{subs:algorithm_configuration}
and successively remove pin sparsification (indicated by suffix $-$S), community-aware coarsening ($-$CAC), and flow-based refinement ($-$F), yielding successively weaker variants. All experiments are performed on benchmark set B.

\paragraph{Results}
Figure~\ref{fig:removal_k_perf_running_times} shows that community-aware coarsening and flow-based refinement substantially improve solution quality at the cost of an increased running time, while pin-sparsification improves running time on some instances at the cost of small quality losses.
In addition to the direct single-shot comparison, we perform \emph{effectiveness tests} using virtual instances.

\paragraph{Virtual Instances}
Since weaker configurations run faster, we create a setting in which each configuration has approximately the same amount of time to compute a $k$-way partition.
More precisely, we use the concept of \emph{virtual instances} \cite{AkhremtsevSS20} to allow the faster configuration to perform additional repetitions.
Given the results of $r$ repetitions of two algorithm configurations $A$ and $B$ for one instance $I$, i.e., a $k$-way partition of a hypergraph $H$, a virtual instance
is computed as follows: First, we choose one repetition of both algorithms uniformly at random. Let $t_1^A$ and $t_1^B$ be the running times of configuration $A$ and
configuration $B$ for that particular repetition, and assume without loss of generality that $t_1^A \geq t_1^B$. We now sample additional repetitions for algorithm $B$ (without
replacement) until the total running time of all sampled repetitions exceeds $t_1^A$, i.e., if the last sample $t_\ell^B$ of algorithm $B$ would exceed $t_1^A$,
it is accepted with probability $p=(t_1^A - \sum_{i \leq i < \ell}t_i^B)/t_\ell^B$. It has been shown that using this approach, the expected running time of the sampled
repetitions of configuration $B$ is the same as the running time of a single repetition of configuration $A$~\cite[Thm.~4.1]{ThesisYaroslav}.
The solution quality of configuration $A$ then corresponds to the quality of the single repetition, while the quality of configuration $B$ is the best
result of all sampled repetitions.
For each of the $1148$ actual instances, we perform $10$ repetitions per configuration. Similar to Ref.~\cite{ThesisYaroslav}, this data is then used to
create 20 virtual instances for each of the $1148$ actual instances -- resulting in a total of $\numprint{22960}$ virtual instances per configuration.\footnote{Only for very
  few instances, slightly more than 10 repetitions would have been needed for the faster configuration. In these cases, we
  restrict the results of the faster configuration to the \emph{best} of the 10 available repetitions, since the impact of those instances on the overall result is considered negligible.}

\begin{figure}[tb!]
  \centering
  \begin{minipage}{.49\textwidth}
  \centering
  \ifpdfplots
    \includegraphics{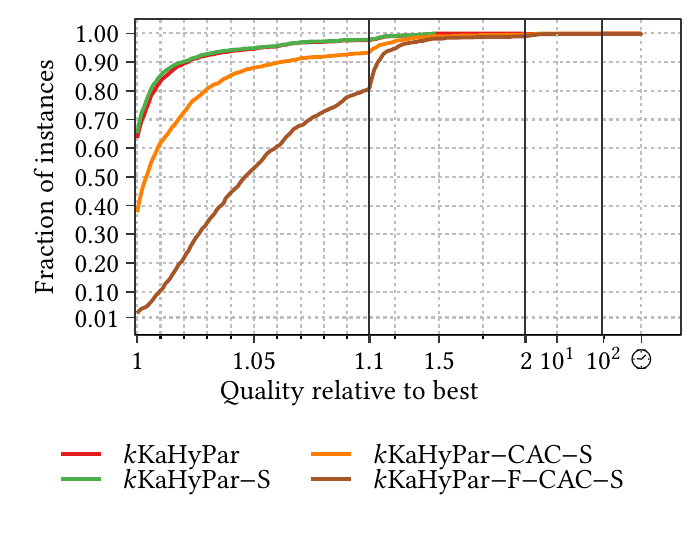}
  \fi

  \end{minipage}
  \begin{minipage}{.49\textwidth}
  \centering
  \vspace{-0.8cm}
  \ifpdfplots
    \includegraphics{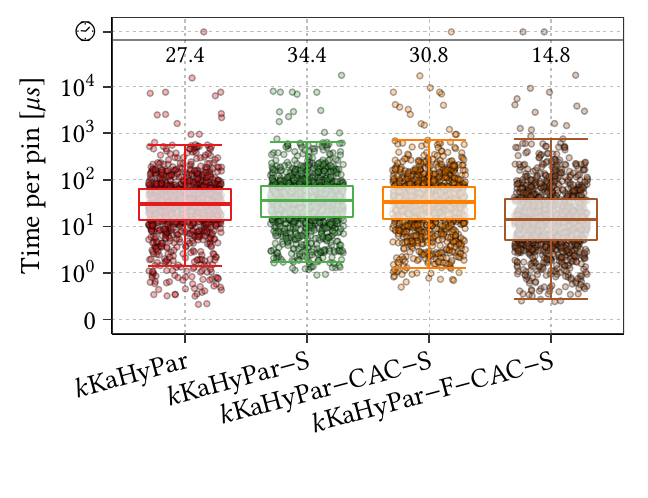}
  \fi

  \end{minipage}
  \vspace{-1cm}
  \caption[Performance profile and running times of KaHyPar configurations using \emph{actual} instances.]
  {Performance profile and running times of KaHyPar configurations using \emph{actual} instances.
  The geometric mean running times per pin are shown above each box plot.  Experiments are performed on benchmark set B.
  Legend: KaHyPar (baseline), KaHyPar$-$S (without pin sparsification),
  KaHyPar$-$CAC$-$S (without community-aware coarsening and pin sparsification),
  KaHyPar$-$F$-$CAC$-$S (without flow-based refinement, community-aware coarsening and pin sparsification).}
  \label{fig:removal_k_perf_running_times}
\end{figure}

The results of the effectiveness tests on virtual instances for KaHyPar are shown in Figure~\ref{fig:removal_k_eff}.
The advanced configurations are always more effective than weaker configurations. Neither KaHyPar$-$F nor KaHyPar$-$F$-$CAC is able to outperform the respective stronger configuration if given the same amount of time. These results substantiate the significant performance differences
shown in the performance profiles
on the \emph{actual} instances in Figure~\ref{fig:removal_k_perf_running_times}.

\begin{figure}[t!]
  \centering
  \hspace{-1cm}
  \begin{minipage}{.32\textwidth}
  \ifpdfplots
    \includegraphics{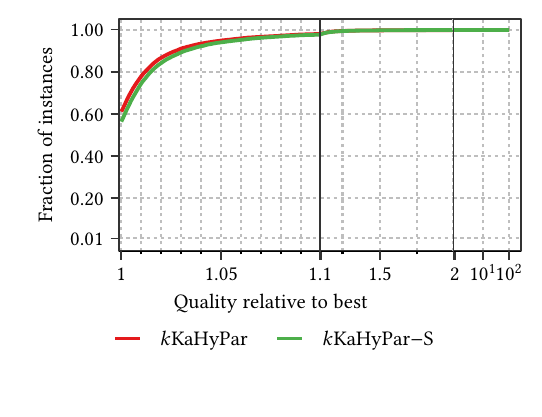}
  \fi

  \vspace{-1cm}
  \end{minipage}
  \hspace{0.25cm}
  \begin{minipage}{.32\textwidth}
  \ifpdfplots
    \includegraphics{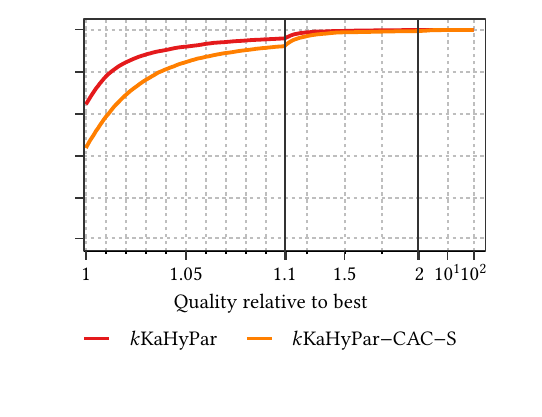}
  \fi

  \vspace{-1cm}
  \end{minipage}
  \begin{minipage}{.32\textwidth}
  \ifpdfplots
    \includegraphics{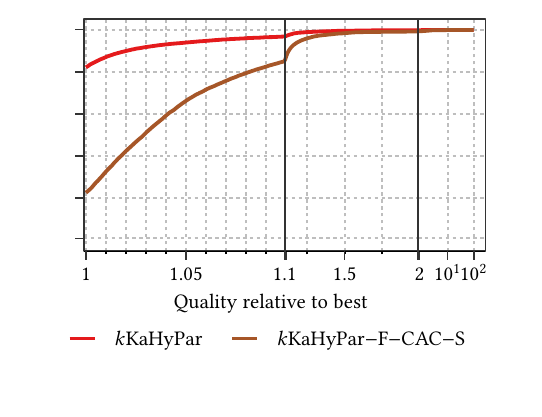}
  \fi

  \vspace{-1cm}
  \end{minipage}
  \caption[Effectiveness tests for KaHyPar using virtual instances.]
  {Effectiveness tests  on benchmark set B for different KaHyPar configurations using \emph{virtual instances}.
  Legend: KaHyPar (baseline), KaHyPar$-$S (without pin sparsification),
  KaHyPar$-$CAC$-$S (without community-aware coarsening and pin sparsification),
  KaHyPar$-$F$-$CAC$-$S (without flow-based refinement, community-aware coarsening and pin sparsification).}
  \label{fig:removal_k_eff}
 \end{figure}

\subsection{Comparison with Other Systems}
\label{sec:partitioner_comparison}

After discussing our reasoning for choosing a set of \emph{seven} state-of-the-art HGP algorithms as competitors,
we compare the performance of KaHyPar with those systems
-- both in terms of solution quality and running time.
As we will see, the differences in running time between the algorithms can be up to several orders of magnitude.
In Section~\ref{sec:repeated_executions}, we therefore evaluate a subset of the best performing algorithms in
a setting where each algorithm is given the same -- large -- amount of time to compute a solution for each instance.
Furthermore, we use this section to demonstrate the effectiveness of our memetic algorithm presented in Section~\ref{subs:memetic_algo}.

\paragraph{Partitioning Systems}\label{sec:choosing_competitors}
We compare KaHyPar with the $k$-way (hMETIS-K) and recursive bipartitioning variants (hMETIS-R)
of hMETIS 2.0 (p1)~\cite{hMETIS-Software},
PaToH 3.2~\cite{PaToH-Software} using both the default (PaToH-D) and the quality preset (PaToH-Q), Zoltan-AlgD~\cite{ZoltanAlgD-Software},
Mondriaan version 4.2.1~\cite{Mondriaan-Software}, and HYPE~\cite{HYPE-Software}.
We choose these tools for the following reasons: PaToH produces better quality than Zoltan's native parallel hypergraph partitioner (PHG) in sequential mode~\cite{ZoltanUserGuide,DBLP:conf/ipps/DevineBHBC06}.
The publicly available version of Parkway~\cite{DBLP:journals/jpdc/TrifunovicK08} on GitHub crashes on most of our instances, but was found to be comparable to Zoltan's PHG~\cite{DBLP:conf/ipps/DevineBHBC06}.
The algebraic distance-based coarsening algorithm of Zoltan-AlgD has been shown to improve the performance of Zoltan's PHG in sequential mode~\cite{shaydulin_et_al:LIPIcs:2018:8937,DBLP:journals/corr/abs-1802-09610,doi:10.1137/17M1152735,EdgePartitioning}.
MLPart is restricted to bipartitioning~\cite{DBLP:conf/aspdac/CaldwellKM00,DBLP:conf/ispd/CongRX03} and was outperformed by both hMETIS~\cite{MLPartVshMetis}
and PaToH~\cite{PaToHvsMLPart}.
The performance of SHP~\cite{DBLP:journals/pvldb/KabiljoKPPSAP17} is deemed comparable to the performance of Zoltan and Mondriaan~\cite{DBLP:journals/pvldb/KabiljoKPPSAP17}.
UMPa~\cite{DBLP:conf/dimacs/CatalyurekDKU12} does not improve on PaToH when optimizing single objective functions that do not benefit from the directed hypergraph model~\cite{DBLP:journals/jpdc/DeveciKUC15}.
Furthermore, kPaToH~\cite{DBLP:journals/jpdc/AykanatCU08} did not perform better than PaToH in preliminary experiments~\cite{KaHyPar-K}.
We exclude our recent parallel algorithms~\cite{GHSS21, GHSS21b} since their components are parallel versions of the components in KaHyPar (though flow-based refinement is still missing). Here, we focus on evaluating these components in a sequential setting and comparing them with other sequential codes.

  \paragraph{Repetitions and Time Limit.}
  We perform ten repetitions with different seeds for all algorithms except for HYPE which is not randomized and hence
  always computes the same results and thus uses a single repetition. Each partitioner had a time limit of eight hours \emph{per} instance and seed.

  \subsubsection{Solution Quality}
  Figure~\ref{fig:km1_perf_running_time_overview} summarizes the results.
  The performance profile plot on the left shows that KaHyPar outperforms the competing algorithms by a large margin,
  as it computes the best partitions for $68.4\%$ of all benchmark instances, and its solution quality is within a factor of $1.1$ of the best algorithm in $94\%$ of all cases.

  Comparing the performance profiles of PaToH-Q with the performance profile of Mondriaan, we can see that PaToH-Q performs better than Mondriaan.
  The performance difference of PaToH-D and Mondriaan is also statistically significant (result of the Wilcoxon signed ranked test is $|Z| = 7.1911$ and $p \le 6.428e-13$), which confirms
  the results of previous studies that suggested that Mondriaan's hypergraph partitioner can be seen as inferior to PaToH~\cite{bisseling2012two,PaToHBetterThanMondriaan}.
  Moreover, we note that
  HYPE (the only non-multi-level algorithm) performs considerably worse than the multi-level systems. This echoes the intuition that by providing a more
  global view of the partitioning problem on coarser levels, multi-level approaches enable local search algorithms to explore local solution spaces very effectively.
  While the performance profile of hMETIS-R is within a factor of $1.1$ of the best algorithm for more than $76.7\%$ of all instances, there are some
  instances for which it performs significantly worse than the best. Although less solutions of PaToH-Q and PaToH-D are within a factor
  $1.1$ from the best, the performance profiles indicate that the worst quality ratios of PaToH-Q and PaToH-D are smaller than those of hMETIS-R.
  Note that the same effect is also visible when comparing Zoltan-AlgD with PaToH-D or Mondriaan.

\begin{figure}[tb!]
  \centering
  \begin{minipage}{.49\textwidth}
    \centering
  \ifpdfplots
    \includegraphics{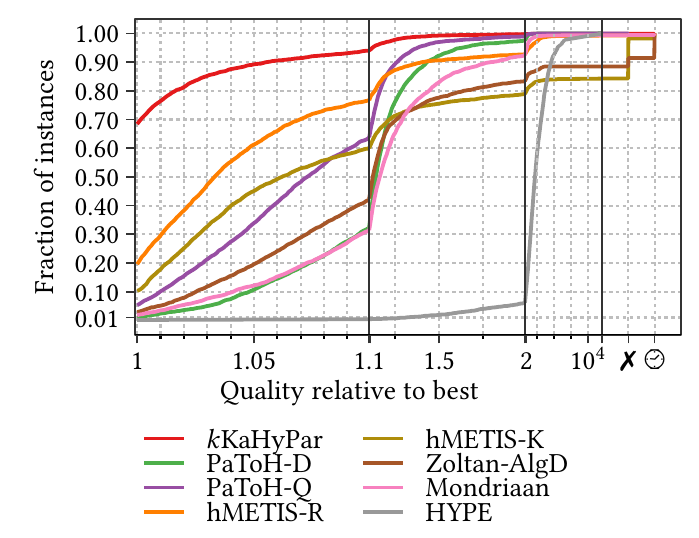}
  \fi

  \vspace{-0.5cm}
  \end{minipage}
  \begin{minipage}{.49\textwidth}
  \centering
  \vspace{-0.8cm}
  \ifpdfplots
    \includegraphics{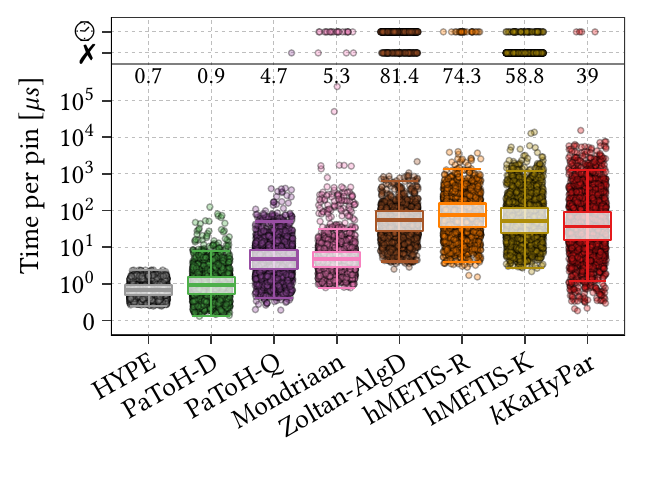}
  \fi

  \vspace{-0.5cm}
  \end{minipage}
  \caption[Performance profile and running times comparing KaHyPar with other partitioners for connectivity optimization.]
  {Performance profile (left) and running times (right) comparing KaHyPar with other partitioners for connectivity optimization.
  The geometric mean running times per pin are shown above each box plot.
  }
  \label{fig:km1_perf_running_time_overview}
\end{figure}

\subsubsection{Running Time}

The running times are shown in Figure~\ref{fig:km1_perf_running_time_overview} (right). We see that although
being based on the $n$-level paradigm and employing more complex techniques such as community-aware coarsening
and flow-based refinements, the running times of KaHyPar are comparable to the running times of both
hMETIS configurations. Furthermore, KaHyPar is faster than Zoltan-AlgD on average.
Note that for both PaToH-Q and PaToH-D, as well as for Mondriaan and HYPE, the median running time is more than an
order of magnitude smaller than the median running times of the other multi-level systems.

\subsubsection{Infeasible Results}
Looking at infeasible solutions, we see that Zoltan-AlgD computes imbalanced solutions for around $2.9\%$ of all instances and that more than $14\%$
of all partitions computed by hMETIS-K are imbalanced. The fact that hMETIS-K often produces imbalanced partitions was also
observed in the graph partitioning experiments of \citet[p.~128]{DBLP:phd/dnb/Schulz13a}. A possible explanation for this behavior could be
the fact that, according to the related publications~\cite{hMETISKTR,DBLP:journals/vlsi/KarypisK00},
hMETIS-K does not limit the maximal vertex weight during the coarsening phase (neither indirectly via a penalty factor in the
rating function nor directly via hard weight constraints). This, in turn, could lead to many heavy vertices at the coarser levels of the hierarchy,
which makes it harder for the initial partitioning algorithms to compute balanced partitions.

For KaHyPar ($4$), hMetis-R ($26$), hMetis-K ($60$), Mondriaan ($19$) and Zoltan-AlgD ($294$), there are some instances that could not be partitioned withing the given time limit of eight hours.
Additionally, Mondriaan aborted partitioning for 4 instances in total.

\subsubsection{The Time/Quality Trade-Off}
The previous section showed that KaHyPar computes superior solutions for most instances, however at the cost of higher running times than, for example, PaToH.
Figure~\ref{fig:km1_cut_tradeoff} concisely summarizes the trade-off between solution quality and running time for connectivity optimization on a per-instance basis. The plot shows one point
for each instance. The $x$-axis gives
the running time ratio $x=K/T$ where $K$ is the running time of KaHyPar and $T$ is the running time of the compared algorithm.
Similarly, the $y$ axis shows $y=k/c-1$ where $k$ is the connectivity value obtained by KaHyPar
and $c$ is the value obtained by the compared algorithm.
Points above zero correspond to instances where the solution of the respective partitioner was better than the
solution of KaHyPar, while for points below zero KaHyPar produced solutions of higher quality.
Imbalanced and thus infeasible solutions are plotted at the $y$ position labeled with \ding{55}.

KaHyPar seems to be the method of choice for high-quality partitioning as it performs better than hMETIS-R, hMETIS-K, and Zoltan-AlgD -- computing solutions of higher quality in a comparable amount of time for most instances.
If running time is more important than solution quality, PaToH offers the best trade-off out of the algorithms in this evaluation, however it is outperformed in both quality and running time by one of our recent parallel algorithms~\cite{GHSS21}.

 \begin{figure}[tb!]
  \centering
  \begin{minipage}{\textwidth}
  \ifpdfplots
    \includegraphics{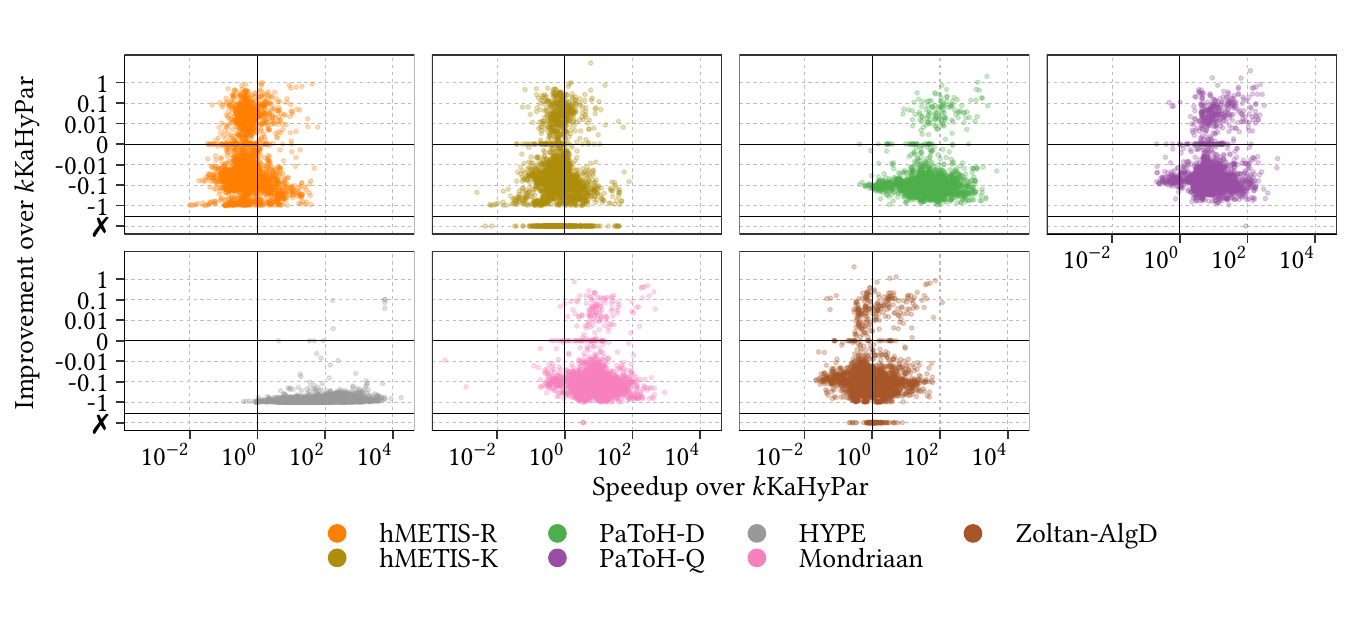}
  \fi

  \vspace{-0.5cm}
  \end{minipage}
  \caption[Visualization of the trade-off between running time and solution quality for connectivity optimization
  and cut-net optimization.]{Visualization of the trade-off between running time and solution quality for connectivity optimization (top)
  and cut-net optimization (bottom). The values of all algorithms are relative to KaHyPar.}\label{fig:km1_cut_tradeoff}
\end{figure}

\subsubsection{Memetic Algorithm and Repeated Executions}\label{sec:repeated_executions}
\paragraph{Motivation}
In the experiments presented in the previous section, each partitioning algorithm was executed the same number of times for each instance
(i.e., ten times with different random seeds). However, we have seen that the difference in running time between algorithms can be
up to three orders of magnitude. In this section, we therefore investigate the partitioning performance in a setting where each
algorithm is given the \emph{same} fairly large amount of time to partition each instance -- thus trying to answer the question whether multiple
repetitions of a very fast algorithm can yield similar or even better results than few repetitions of a slower, more advanced algorithm.
Furthermore, this evaluation is relevant for tasks such as  application-specific integrated circuit (ASIC) design, where one can afford to run the
partitioner for hours or days, since it will take weeks to create the final implementation~\cite{DHauBo97}.
The memetic algorithm from Section~\ref{subs:memetic_algo} is particularly relevant in this context.

\paragraph{Methodology}
The experiments in this section were done on benchmark set C.
For each instance, every partitioning algorithm was given
\emph{eight} hours time to compute a solution. We performed five repetitions with different seeds for each test
instance and algorithm. Due to the large amount of computing time necessary to perform these experiments, we always partitioned $16$ instances in
parallel on a single node of the compute cluster.\footnote{Compared to partitioning a single instance on a single node, we did not observe considerably different results.}
For simplicity, we refer to the memetic algorithm as KaHyPar-E in this section.
While all non-evolutionary algorithms \emph{repeatedly partition} each instance until the time limit is reached, KaHyPar-E evolves a population
of solutions.

\paragraph{Experimental Evaluation}
The performance profiles in Figure~\ref{fig:convergence_partitioning} summarize the experimental results.
The plot on the left is based on the \emph{best} solutions computed by all algorithms after repeateldy partitioning each instance for eight hours.
We see that KaHyPar-E is able to effectively explore the global solution space
 -- computing the best partitions for $94.6\%$ of all instances while never being more than a factor of $1.062$ worse than the best algorithm.
 However, even in this setting (and even competing with KaHyPar-E), KaHyPar computes the best solutions for $20.6\%$ of the instances.

Furthermore, we observe that PaToH-Q seems to compare much more favorably with hMETIS than in the previous single-shot comparison.
However, the Wilcoxon signed ranked test reveals that hMetis-R produces statistically significant better partitions than
PaToH-Q ($|Z| = 6.9935,p \le 2.68e-12$), but PaToH-Q outperforms hMetis-K ($|Z| = 4.7788,p \le 1.763e-06$).
The results clearly indicate that it does not suffice to use the best solutions of repeated executions of a faster partitioner
to achieve the same solution quality as KaHyPar or KaHyPar-E.

 \begin{figure}[t!]
  \centering
  \hspace{-0.75cm}
  \begin{minipage}{.49\textwidth}
  \ifpdfplots
    \includegraphics{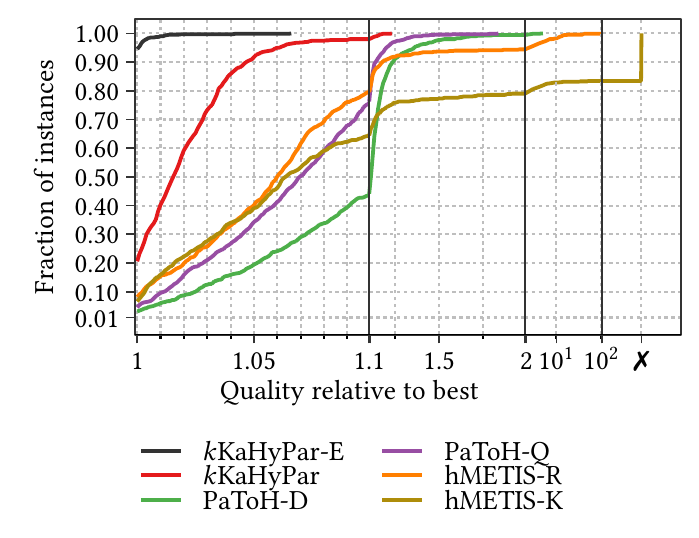}
  \fi

  \vspace{-0.5cm}
  \end{minipage}
  \begin{minipage}{.49\textwidth}
  \ifpdfplots
    \includegraphics{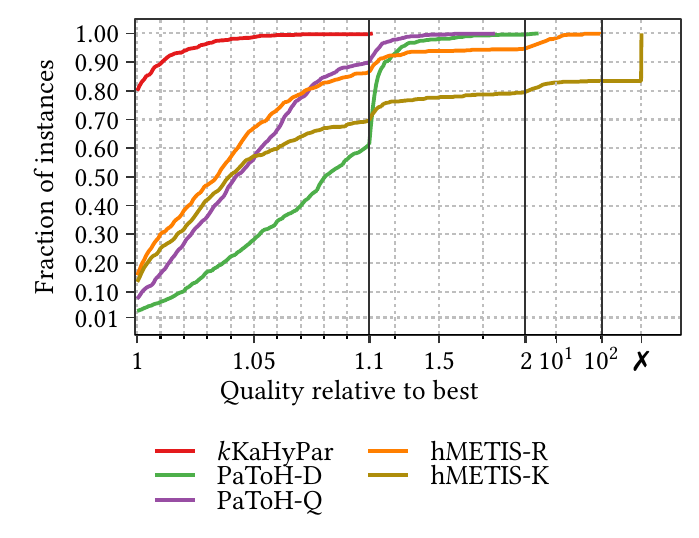}
  \fi

  \vspace{-0.5cm}
  \end{minipage}
  \caption[Performance profile comparing KaHyPar and KaHyPar-E with other partitioners.]{Performance profile comparing KaHyPar and KaHyPar-E with other partitioners on benchmark set C. The left plot uses the \emph{best} result that each system computed after partitioning
  each instance for eight hours. The right plot compares the very first results of KaHyPar (reported during the eight hour time period) to the \emph{best} results produced
  by the other algorithms.}\label{fig:convergence_partitioning}
\end{figure}

 In Figure~\ref{fig:convergence_partitioning} (right), we strengthen this argument for KaHyPar by comparing the \emph{best} solutions found by the competing algorithms with
 repeated partitioning to the solutions of \emph{single} partitioning calls to KaHyPar
 (i.e., the \emph{first} results reported for each instance and seed during the eight hour time period).
 In this setting, KaHyPar still performs better than the other partitioning systems -- computing the best solutions for around $80\%$ of all benchmark instances.

 \subsection{Case Study: Graph Edge Partitioning}\label{sec:graph_edge_partitioning}
\paragraph{Motivation}
Traditional node-based graph partitioning has been essential for making distributed graph algorithms efficient
in the Think-Like-A-Vertex model of computation~\cite{mccune-tlav-2015}. In this model, node-centric operations are performed
in parallel by mapping nodes to processing elements (PEs) and executing node computations in parallel. Nearly all algorithms in this model require information
to be communicated between neighbors -- which results in network communication if stored on different PEs -- and therefore high-quality
graph partitioning directly translates into less communication and faster overall running time.
However, node-centric computations have serious shortcomings on power-law graphs, which have a skewed
degree distribution. In such networks, the overall running time is negatively affected by nodes with very high degrees,
which can result in more communication steps. To combat these effects, \citet{gonzalez-powergraph-2012}
introduced edge-centric computations, which duplicate node-centric computations across edges to reduce
communication overhead. In this model, \emph{edge partitioning} -- partitioning the edge set of a graph into
roughly equally-sized blocks while minimizing node replications -- must be used to reduce the overall running time.
However, like node-based partitioning, edge partitioning is NP-hard~\cite{bourse-2014}.
Since the problem can be solved directly via hypergraph partitioning, we use it as a case study to demonstrate
the performance of KaHyPar on instances that were \emph{never} used during the development or the tuning of the framework's algorithmic components.

\paragraph{The Edge Partitioning Problem}
Let $G=(V,E,c,\omega)$ be an undirected, weighted graph.
Similar to the node partitioning problem, the \emph{edge partitioning problem} asks for
blocks of edges $E_1, \dots, E_k$ that partition $E$, i.e., $E_1 \cup \dots \cup E_k = E$, $E_i \neq \emptyset$ for $1\leq i \leq k$, and $E_i \cap E_j = \emptyset$ for $i \neq j$.
The \emph{balance constraint} demands that $\forall i \in \{1..k\}: \omega(E_i) \leq (1 + \varepsilon) \ceil{\frac{\omega(E)}{k}}$.
The objective is to minimize the \emph{vertex cut} $\sum_{v \in V} |VC(v)| - 1$ where $VC(v) := \{ i : \incnets(v) \cap E_i \neq \emptyset \}$.
Intuitively, the objective expresses the number of required \emph{replicas} of nodes:
If a node $v$ has to be copied to each block that has edges incident to $v$,
the number of replicas of that node is $|VC(v)| - 1$.

Li et al.~\cite{li2017spac} noted that an edge partition of a graph $G$ can be computed by transforming $G$ into a hypergraph $H$,
 partitioning $H$ into $k$ blocks while optimizing the connectivity metric  $\ocon$, and then using the hypergraph partition to infer an edge
partition of $G$. The hypergraph contains a vertex for each graph edge $e \in E$, and a hyperedge for each graph node $v \in V$, which contains the
graph edges to which the corresponding node is incident. An example of this approach is shown in Figure~\ref{fig:edge_partitioning_example}.

 \begin{figure}[t!]
   \centering
   \includegraphics[width=0.8\textwidth]{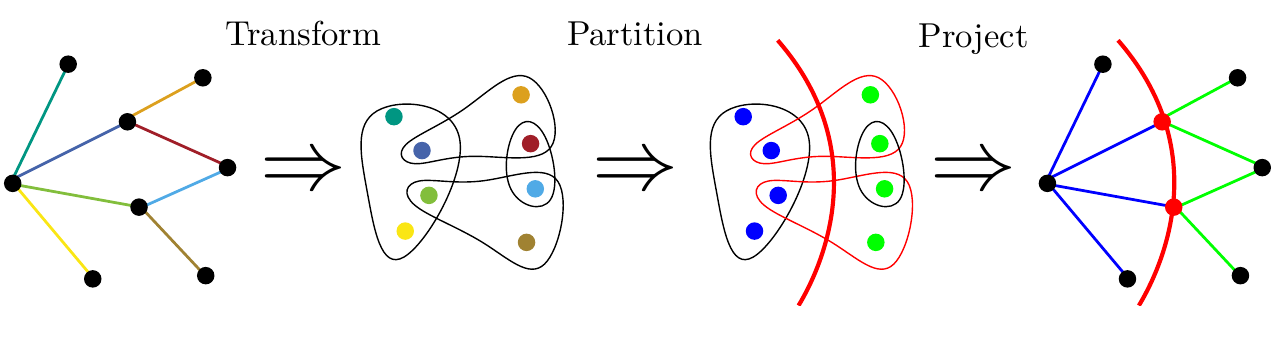}
   \caption[The edge partitioning approach of Li et al.~\cite{li2017spac}.]{Computing an edge partition via hypergraph partitioning as proposed by Li et al.~\cite{li2017spac}. After transforming the graph into a hypergraph and computing
a connectivity-optimized hypergraph partition, the hypergraph solution induces a partition of the edge set of the graph.}\label{fig:edge_partitioning_example}
\end{figure}

In our conference paper~\cite{EdgePartitioning}, we present a fast parallel auxiliary graph construction algorithm in the \emph{distributed} setting,
that -- combined with advanced parallel node partitioning algorithms -- yields high-quality edge partitions in a scalable way.
Here, we restrict our focus to the usage of \emph{sequential} hypergraph partitioning algorithms to solve the edge partitioning problem directly.

\paragraph{Instances and Methodology.}
In the evaluation, we use a benchmark set of
46 hypergraphs (set D) which are derived from a set of benchmark graphs as described above. More precisely, we
use all instances of the Walshaw standard graph partitioning benchmark~\cite{DBLP:journals/jgo/SoperWC04},
SPMV graphs~\cite{li2017spac}, and random hyperbolic rhgX graphs. SPMV graphs are
bipartite locality graphs for sparse matrix vector multiplication (SPMV), which were
used in the evaluation of Li et al.~\cite{li2017spac}. Given an $n\times n$ matrix $M$ (in our case the
adjacency matrix of the corresponding graph), an SPMV graph corresponding to an
SPMV computation $Mx = y$ consists of $2n$ vertices representing the $x_i$ and $y_i$ vector
entries and contains an edge $(x_i , y_j )$ if $x_i$ contributes to the computation of $y_j$ , i.e., if
$M_{ij} \neq 0$. The rhgX graphs were chosen since their degree distributions follow a power
law (and they are thus targeted by edge partitioning techniques). They are generated
using KaGen~\cite{kagen} with a power law exponent of $2.2$ and an average degree of $8$.

For each instance and each algorithm (except HYPE), we perform five repetitions with different seeds. As before,
for HYPE, we report the results of one iteration using the default configuration, since employing randomization did not improve solution quality~\cite{KaHyPar-MF-JEA}.

\paragraph{Results}
The results of our experiments are summarized in Figure~\ref{fig:edge_partitioning_results}. Considering the performance
profile plot in Figure~\ref{fig:edge_partitioning_results} (left), we see that out of all partitioning algorithms KaHyPar again performs best -- computing the best edge partitions
on $82.3\%$ of all benchmark instances. KaHyPar is never more than a factor of $1.13$ worse than the best algorithm.
Interestingly, and in contrast to the results presented in Section~\ref{sec:partitioner_comparison}, the performance difference between hMETIS-R
and its direct $k$-way counterpart hMETIS-K is small in this case study. This could be seen as an indication that, whenever
hMETIS-K is able to compute feasible solutions, its solution quality could be comparable to that of hMETIS-R.

\begin{figure}[t!]
  \centering
  \begin{minipage}{.49\textwidth}
  \centering
  \ifpdfplots
    \includegraphics{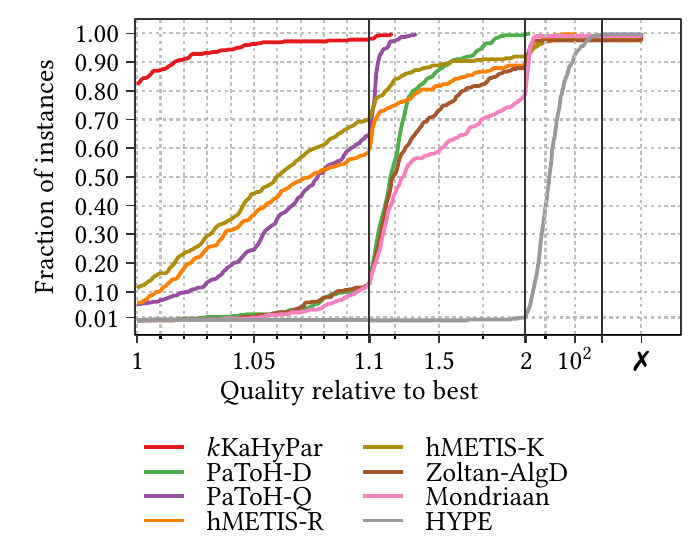}
  \fi

  \vspace{-0.5cm}
  \end{minipage}
  \begin{minipage}{.49\textwidth}
  \centering
  \vspace{-.81cm}
  \ifpdfplots
    \includegraphics{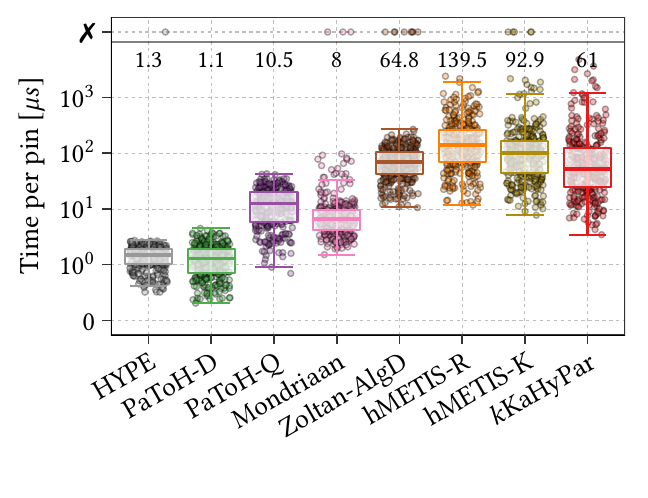}
  \fi

  \vspace{-0.5cm}
  \end{minipage}
  \caption[Solution quality and running times for edge partitioning.]
  {Comparing the results of all algorithms for the edge partitioning experiments on benchmark set D: solution quality (left) and running time (right).}
  \label{fig:edge_partitioning_results}
\end{figure}

The Wilcoxon signed ranked test indicates that the difference between both hMETIS configurations and PaToH-Q is
not statistically significant ($|Z| = 0.4833, p = 0.6289$ for hMetis-R and $|Z| = 1.2917, p = 0.1965$ for hMetis-K).
Figure~\ref{fig:edge_partitioning_results} (right) shows that
the running times of both KaHyPar configurations are again comparable to that of hMETIS-K, hMETIS-R, and Zoltan-AlgD.
As before, both PaToH configurations, as well as HYPE and Mondriaan are considerably faster than the other partitioning tools.

\paragraph{Concluding Remarks}
The experimental results yield similar conclusions as the experiments on benchmark set A.
Since benchmark set D was not used during the development of KaHyPar, this
provides further evidence of our claim that KaHyPar can be seen as the state of the art
for high-quality hypergraph partitioning.

\subsection{Case Study: Traditional Graph Partitioning}\label{sec:graph_vertex_partitioning}
\paragraph{Motivation.}
Since hypergraph partitioning is a generalization of graph partitioning, we now compare KaHyPar with the graph partitioner KaFFPa~\cite{kaffpa} from the KaHIP framework~\cite{sandersschulz2013} in order to evaluate how good our algorithm performs for traditional graph partitioning tasks.
We use both the \texttt{strong} (KaFFPa-Strong) and the \texttt{strongsocial} (KaFFPa-StrongS) configurations, and use the improved flow network described in Ref.~\cite{KaHyPar-MF-JEA}.
The former configuration is designed for high-quality partitions of mesh graphs, while the latter is targeted at partitioning complex networks such as web graphs and social networks.
Note that KaFFPa achieved the highest quality partitions in the DIMACS challenge on graph partitioning~\cite{DBLP:conf/dimacs/2012}.
Hence, we compare KaHyPar to a state-of-the-art system that uses similar techniques (except the preprocessing and the $n$-level approach).

\paragraph{Instances and Methodology.}
We use the 21 large web graphs and
social networks used in the work of Meyerhenke et al.~\cite{DBLP:conf/wea/MeyerhenkeSS14} (benchmark set E),
and the graphs used in the final evaluation of the 10th DIMACS Implementation
Challenge on Graph Partitioning and Graph Clustering~\cite{DBLP:conf/dimacs/2012} (benchmark set F).
We excluded graph uk-2007-05 from benchmark set F, because no algorithm involved
in our experimental evaluation was able to partition that graph on our system.
In order to be consistent with the experiments presented in the previous sections, we impose a time limit of eight hours per instance for each algorithm.

\paragraph{Comparison with KaHyPar -- Complex Networks.}
The experimental results for the web graphs and social networks of benchmark set E are summarized in Figure~\ref{fig:snw_results}.
Considering the performance profiles in Figure~\ref{fig:snw_results} (left), we see that KaHyPar computes the best solutions for $58.5\%$ of all instances.
It is followed by KaFFPa-StrongS ($29.3\%$), and KaFFPa-Strong ($12.2\%$).
Figure~\ref{fig:snw_results} (right) shows that KaHyPar is slightly faster on average than
KaFFPa-StrongS, which could not partition one instance within the time limit.

\begin{figure}[t!]
  \centering
  \begin{minipage}{.49\textwidth}
  \ifpdfplots
    \includegraphics{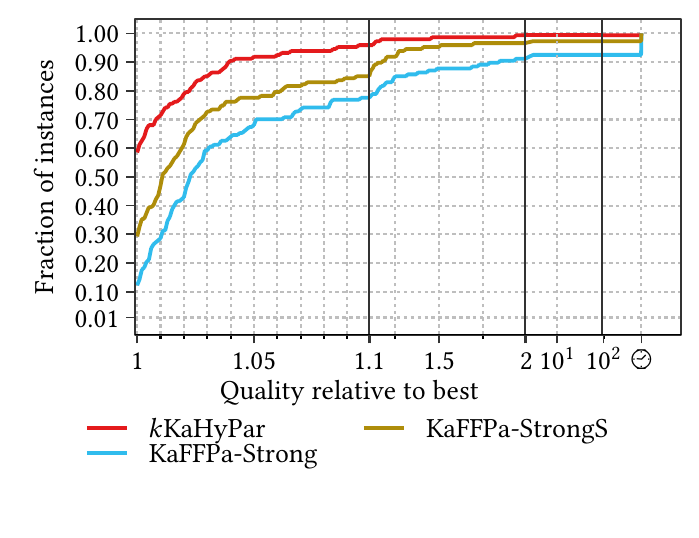}
  \fi

  \vspace{-1cm}
  \end{minipage}
  \begin{minipage}{.49\textwidth}
  \vspace{-0.1cm}
  \ifpdfplots
    \includegraphics{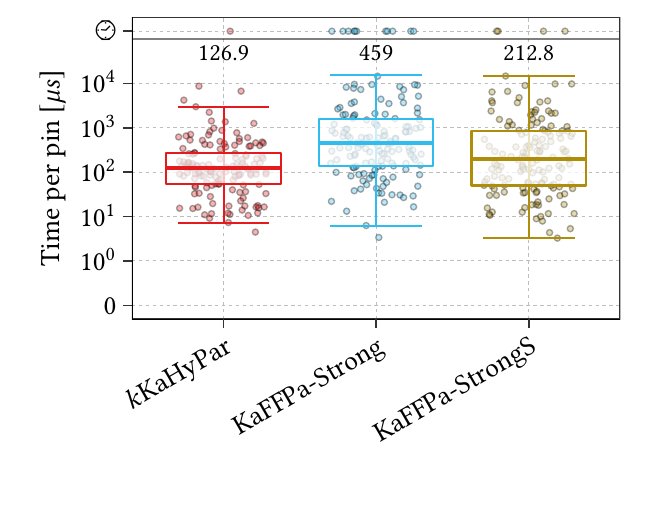}
  \fi

  \vspace{-1cm}
  \end{minipage}
  \caption[Solution quality and running times of KaHyPar, KaFFPa-Strong, and KaFFPa-StrongS for benchmark set E.]
  {Solution quality (left) and running times (right) of KaHyPar, KaFFPa-Strong, and KaFFPa-StrongS for benchmark set E (social, web).}\label{fig:snw_results}
\end{figure}

\paragraph{Comparison with KaHyPar -- DIMACS Graphs.}
The results for benchmark set F are summarized in Figure~\ref{fig:dimacs_results}.
The performance profiles depicted in Figure~\ref{fig:dimacs_results} (left) show that none of the algorithms were able to partition all instances,
because the partitioners could not finish within the time limit of eight hours.
We see that KaFFPa-Strong resp.~KaFFPa-StrongS compute the best solutions for $29.4\%$ resp.~$35.3\%$ of all instances, while
KaHyPar computes the best solutions in $37\%$ of all cases.
Looking at Figure~\ref{fig:dimacs_results} (right), we see that KaFFPa-Strong is faster than KaHyPar and KaHyPar is faster than KaFFPa-StrongS.

\begin{figure}[t!]
  \centering
  \begin{minipage}{.49\textwidth}
  \ifpdfplots
    \includegraphics{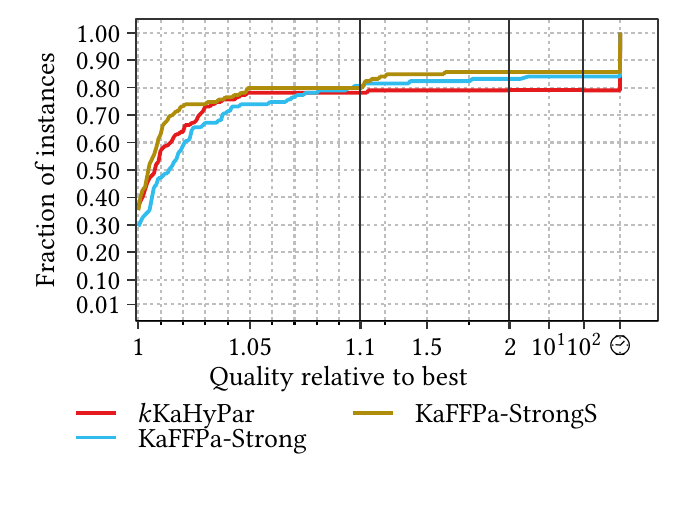}
  \fi

  \vspace{-1cm}
  \end{minipage}
  \begin{minipage}{.49\textwidth}
    \vspace{-0.1cm}
  \ifpdfplots
    \includegraphics{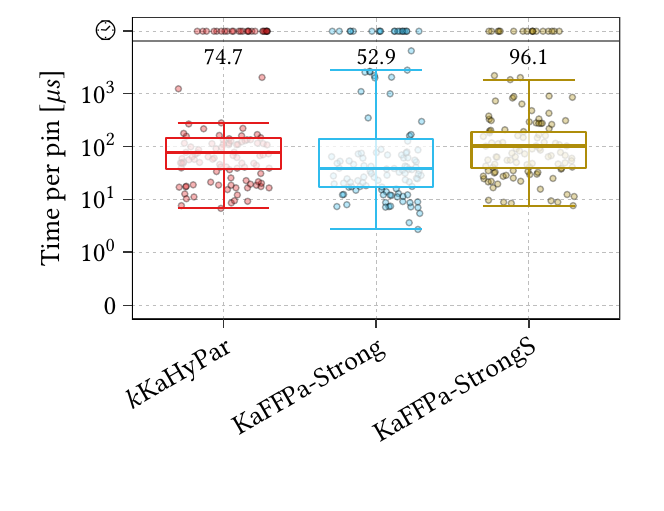}
  \fi

  \vspace{-1cm}
  \end{minipage}
  \caption[Solution quality and running times of KaHyPar, KaFFPa-Strong, and KaFFPa-StrongS for benchmark set F.]
  {Solution quality (left) and running times (right) of KaHyPar, KaFFPa-Strong, and KaFFPa-StrongS for benchmark set F (Dimacs).}\label{fig:dimacs_results}
\end{figure}

 \paragraph{Concluding Remarks}
Given these results, we conclude that KaHyPar is also effective in the context of graph partitioning.
In a comparison with the strongest KaFFPa configurations, it computes solutions of slightly higher quality
for complex networks, and solutions of similar quality for the DIMACS graphs in a comparable amount of time.

\section{Conclusions and Future Work}\label{s:conclusions}
We have shown how a careful combination of different heuristics, e.g.
locality sensitive hashing, Louvain clustering, portfolio-based initial partitioning, localized $k$-way
local search, flow techniques, V-cycles, and memetic
algorithms can be integrated
within the framework of an $n$-level algorithm to make KaHyPar the
currently highest quality hypergraph partitioner. Implementation
techniques like a dynamic hypergraph data structures and lazy update
techniques make this sufficiently fast to process tens of thousands of
pin per second.  This is fast enough for many combinatorial
optimization applications like VLSI design or quantum circuit simulation~\cite{gray2021hyper}.
For applications like sparse matrix multiplication, where improved cuts translate into
relatively small savings in running time, systems like PaToH currently offer a better cost--quality trade-off.
Here, using KaHyPar is only warranted if its cost can be amortized over many iterations.
Thus, an interesting direction for future research is improving performance
while preserving as much as possible of the achieved
quality. \emph{Parallelization} seems like a promising approach to
achieve this. However, in the past, parallel (hyper)graph partitioners
paid a quality penalty for improved performance. We are intensively
working on changing this. Our first step was done for shared-memory \emph{graph}
partitioning \cite{AkhremtsevSS20} where we achieve quality comparable
to a state-of-the art graph partitioner using corresponding components.
Our recently published parallel hypergraph
partitioner Mt-KaHyPar~\cite{GHSS21} outperforms PaToH-Q and is also faster
than PaToH-D when using a moderate number of threads.
We are also working on parallelizing the $n$-level approach. In a first study~\cite{GHSS21b} it achieves
similar quality as KaHyPar without flows and is also an order of magnitude faster with ten threads.
A key remaining step to achieve the quality of KaHyPar is the integration of parallel flow computations.
For very large instances, also a distributed-memory parallelization seems interesting.

A different cost--quality trade-off might be achievable by using a
KaHyPar-like system not on the input hypergraph but on a considerably
contracted hypergraph. Of course, also achieving even higher quality than
KaHyPar will remain interesting. Many improvements are
conceivable, such as integrating machine learning or integer linear programming.
For small imbalance parameters,
generalizing negative cycle detection techniques previously used for
graph partitioning \cite{sandersschulz2013} could be interesting.




\begin{acks}
  The authors would like to thank the additional co-authors of previous conference publications, namely Robin Andre, Michael Hamann, Vitali Henne, Henning Meyerhenke, Daniel Seemaier, Darren Strash, and Dorothea Wagner. This work was partially supported by DFG grants WA654/19-
  2, SA933/10-2, SA933/11-1, SCHU 2567/1-2, and the Leibniz prize for Peter Sanders. The authors acknowledge further support by the state of Baden-Württemberg through bwHPC.
\end{acks}

\bibliographystyle{ACM-Reference-Format}
\bibliography{references}

\appendix

\clearpage
\section{Benchmark Statistics}
\label{app:benchmarks}

\begin{figure*}[!htb]
	\centering
	\vspace{-1cm}
  \ifpdfplots
    \includegraphics{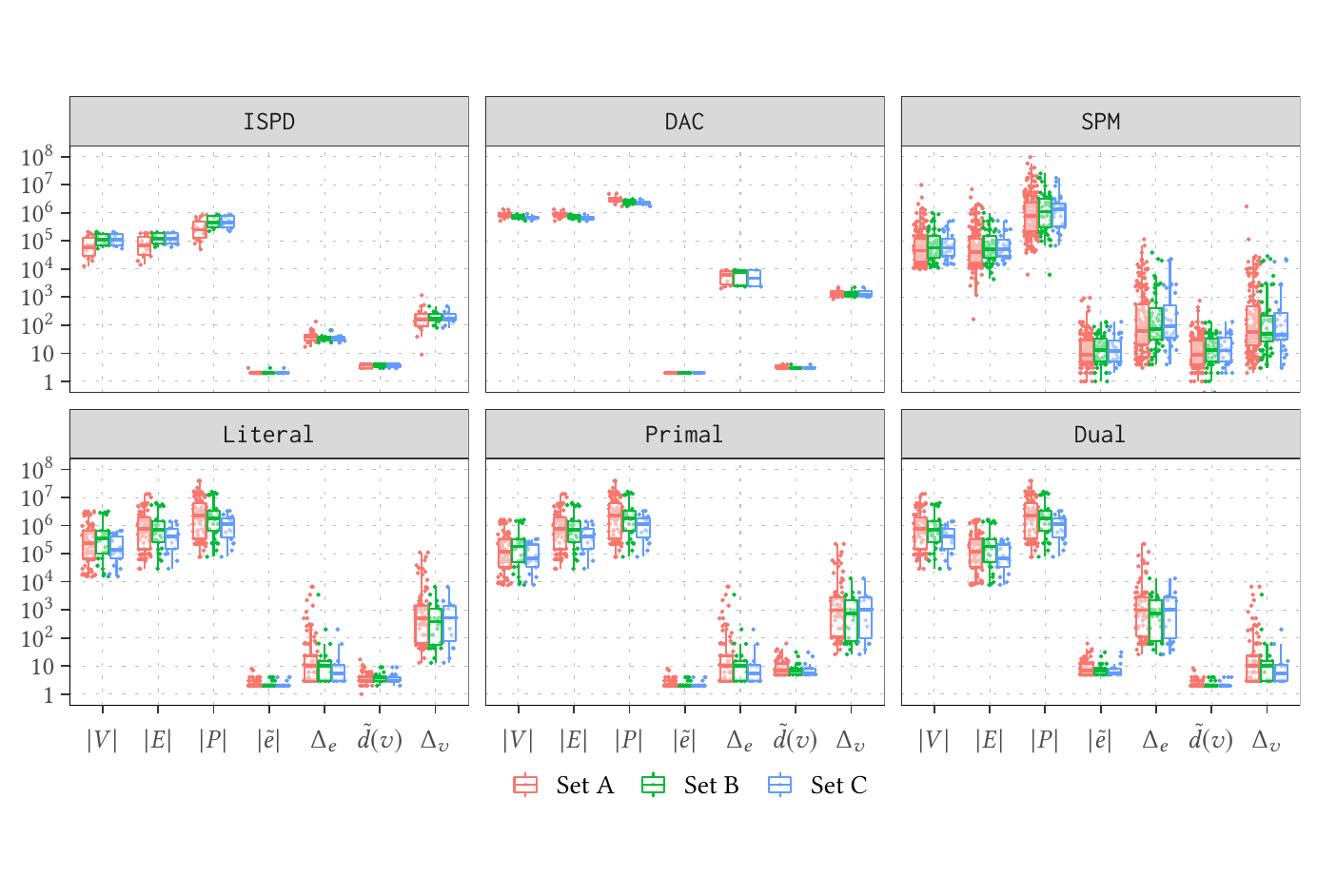}
  \fi

	\vspace{-1.5cm}
  \caption{Summary of different properties for benchmark set A, B and C and the different sources. It shows for each
           hypergraph (points), the number of vertices $|V|$, nets $|E|$ and pins $|P|$, as well as the median and maximum
           net size ($\medsize$ and $\maxsize{e}$ and vertex degree ($\meddeg$ and $\maxsize{v}$).}
	\label{fig:benchmark_set}
\end{figure*}










\end{document}